%% file: arxiv_version_3/main_arxiv.tex
\documentclass[11pt,letterpaper]{article}

\title{Quantum property testing in sparse directed graphs~\footnote{A preliminary version of this work appeared
in Approximation, Randomization, and Combinatorial Optimization. Algorithms and Techniques (APPROX/RANDOM), volume 353, pages 32:1–32:24. 
Schloss Dagstuhl – Leibniz-Zentrum für Informatik, 2025~\cite{AMSS25_random}.}}

\input{arxiv_version_3/macro}

\author[1]{Simon Apers\thanks{Email: \{apers, magniez, szabo\}@irif.fr. Research supported in part by  the ERC Advanced Grant PARQ (grant agreement No 885394), the European QuantERA project QOPT (ERA-NET Cofund 2022-25),
the French PEPR integrated projects EPiQ (ANR-22-PETQ-0007) and HQI (ANR-22-PNCQ-0002),
and the French ANR project QUOPS (ANR-22-CE47-0003-01).}}

\author[1]{Fr\'ed\'eric Magniez\samethanks[2]}

\author[2]{Sayantan Sen\thanks{Email: sayantan789@gmail.com. Research supported by the NRF Investigatorship award (NRF-NRFI10-2024-0006) and CQT Young Researcher Career Development Grant (25-YRCDG-SS).}}

\author[1,3]{D\'aniel Szab\'o\samethanks[2]\thanks{Research supported in part by the German Federal Ministry of Research, Technology and Space (QuSol, 13N17173).}}

\affil[1]{Universit\'e Paris Cit\'e, CNRS, IRIF, Paris, France}
\affil[2]{Centre for Quantum Technologies, National University of Singapore, Singapore}
\affil[3]{Ruhr-Universität Bochum, Bochum, Germany}

\date{}

\begin{document}
\begin{titlepage}

\maketitle
\thispagestyle{empty}

\input{arxiv_version_3/abstract_arxiv}
\end{titlepage}

\input{arxiv_version_3/intro_arxiv}

\input{arxiv_version_3/preliminaries_arxiv}

\input{arxiv_version_3/algorithm_arxiv}

\input{arxiv_version_3/kstarfreeness_test_arxiv}

\input{arxiv_version_3/other_problems_arxiv}

\input{arxiv_version_3/conclusion_arxiv}

\bibliographystyle{alpha}

\bibliography{arxiv_version_3/reference_arxiv}

\appendix

\input{arxiv_version_3/boolean_poly_arxiv}
\input{arxiv_version_3/proofs_appendix_arxiv}

\input{arxiv_version_3/3colorability_appendix_arxiv}

\end{document}

%% file: arxiv_version_3/macro.tex
\usepackage[T1]{fontenc}
\usepackage{amsmath,amsthm,amssymb,xcolor}
\usepackage[margin=1in]{geometry}
\usepackage[utf8]{inputenc}
\usepackage{latexsym,bbm,graphicx,float,mathtools}
\newenvironment{claimproof}[1]{\par\noindent\emph{Proof of claim.}\space#1}{\hfill$\diamond$\medskip}
\usepackage{amsthm}
\usepackage{hyperref}
\usepackage[capitalise]{cleveref}

\usepackage{thm-restate}

\usepackage{apxproof}

\usepackage{authblk}

\usepackage[ruled,linesnumbered,vlined,noend]{algorithm2e}

\usepackage{braket}
\usepackage{stackrel}
\usepackage[all]{xy}

\usepackage{enumitem}
\setlist{itemsep=0em} 

\newcommand{\E}{\mathop{\mathbb{E}}}
\newcommand{\sgn}{\mathrm{sgn}}
\newcommand{\gapcol}{\mathrm{Collision}}
\newcommand{\dgapcol}{\mathrm{dCollision}}
\newcommand{\thr}{\mathrm{THR}}
\newcommand{\gapor}{\mathrm{GapOR}}
\newcommand{\bdeg}{{\mathrm{bdeg}}}
\newcommand{\dpdeg}{{\mathrm{dpdeg}}}

\newcommand{\phd}{\mathrm{phd}}

\newcommand{\cS}{\mathcal{S}}

\newcommand{\cE}{\mathcal{E}}

\newcommand{\cP}{\mathcal{P}}

\newcommand{\cO}{\mathcal{O}}

\newcommand{\eps}{\varepsilon}

\newcounter{mynotes}
\newtheorem{theorem}{Theorem}[section]
\newtheorem{lemma}[theorem]{Lemma}
\newtheorem{claim}[theorem]{Claim}
\newtheorem{proposition}[theorem]{Proposition}
\newtheorem{definition}[theorem]{Definition}
\newtheorem{remark}[theorem]{Remark}

\usepackage{abstract}

\newcommand*\samethanks[1][\value{footnote}]{\footnotemark[#1]}

%% file: arxiv_version_3/abstract_arxiv.tex
\begin{abstract}
We initiate the study of quantum property testing in sparse directed graphs, and more particularly in the unidirectional model, where the algorithm is allowed to query only the outgoing edges of a vertex.
In the classical unidirectional model, the problem of testing $k$-star-freeness, and more generally $k$-source-subgraph-freeness, is almost maximally hard for large $k$. We prove that this problem has almost quadratic advantage in the quantum setting.
Moreover, we show that this advantage is nearly tight, by showing a quantum lower bound using the method of dual polynomials on an intermediate problem for a new, property testing version of the $k$-collision problem that was not studied before.

To illustrate that not all problems in graph property testing admit such a quantum speedup, we consider the problem of $3$-colorability in the related undirected bounded-degree model, when graphs are now undirected. This problem is maximally hard to test classically, and we show that also quantumly it requires a linear number of queries.

\end{abstract}

%% file: arxiv_version_3/intro_arxiv.tex
\section{Introduction}
\subsection{Context}

In this modern big data era, the size of inputs has grown so much that even just reading the full input has become extremely expensive computationally. To tackle this challenge, the framework of \emph{property testing} was initiated. It focuses on designing ultrafast algorithms (also known as ``testers'') that read only a small part of the input, and distinguish inputs that satisfy some property from inputs that are ``far'' from satisfying it.
As a possible use-case, when the exact computation is expensive, one can use property testing algorithms as a precursor to running the final algorithm.
If the input does not pass the property testing test, we can safely reject it, without running the expensive final computation.

At the same time, the field of quantum computing has significantly influenced many computer science paradigms, including cryptography, algorithms, and large-scale data processing. This new perspective on computer science based on quantum physics has sparked many fresh research directions.
This includes the topic of this work, which combines quantum computing and property testing.
More specifically, we consider quantum algorithms for \emph{graph} property testing.

Graphs are of paramount importance for instance when it comes to understanding large datasets, since they provide a natural way to represent and analyze complex relationships inside datasets.
Goldreich, Goldwasser, and Ron~\cite{GGR98} were the first to consider graphs in the context of property testing. 
Formally, given some form of query access to an unknown graph $G$ on $N$ vertices, and a property $\cP$ of interest, the goal is to distinguish with high probability if $G$ satisfies the property $\cP$, or whether it is ``\emph{far}'' from all graphs that satisfy $\cP$, with a suitable notion of farness.
In \cite{GGR98}, the ``dense'' model was considered, where a graph is accessed through \emph{adjacency queries}: for a pair of vertices $(u,v)$, the query reveals whether $(u,v)$ is an edge in the graph.
In this model, a graph $G$ is $\eps$-far from satisfying $\cP$ if one needs to add or remove at least $\eps N^2$ edges of $G$ to obtain a graph that satisfies $\cP$.

In a later work, Goldreich and Ron~\cite{goldreich2002property} introduced the ``bounded-degree'' model for testing sparse graphs, focusing on the properties of bipartiteness and expansion.  
In this model, 
a $d$-bounded degree graph $G$ with $N$ vertices is accessed by performing \emph{neighbor queries}: for a vertex $v$ and an integer $i \in [d]$, the query $(v,i)$ returns either the $i$-th neighbor of $v$, or some special symbol if $v$ has less than $i$ neighbors. 
The graph $G$ is said to be $\eps$-far from some property $\cP$, if one needs to add or delete at least $\eps d N$ edges of $G$ to obtain a graph that satisfies $\cP$.
Over the last two decades, there has been a significant number of works in this model, and we refer the interested reader to the books by Goldreich \cite{goldreich2017introduction} and Bhattacharyya and Yoshida~\cite{bhattacharyya2022property} and several surveys~\cite{Fischer01, DBLP:journals/fttcs/Ron09, Czumaj2010,DBLP:journals/siamdm/RubinfeldS11}.

Some researchers have considered 
efficient quantum algorithms for testing both classical and quantum objects, see for instance \cite{buhrman2008quantum,ambainis2016efficient,harrow2017sequential,ben2020symmetries,DBLP:journals/qic/ApersS19} and the survey~\cite{montanaro2016survey}. 
Notably, the authors in \cite{ambainis2011quantum} initiated the study of bounded degree graph property testing in the quantum model.
One important result in this context is the result of \cite{ben2020symmetries}, who proved that there can be exponential quantum advantage in the bounded degree graph model of property testing. However, as mentioned in their paper, the graph property admitting the exponential quantum advantage is not a natural one. 

\subsection{Property testing of directed bounded degree graphs}

While all of the aforementioned works consider undirected graphs, many real-world instances (such as the world wide web) actually correspond to \emph{directed} graphs.
Consequently, Bender and Ron~\cite{bender2002testing} introduced a model of property testing for directed graphs, focusing on the properties of acyclicity and connectivity.
Following that work, 
we open a new research line by studying quantum algorithms for testing directed graphs.
As we will see, by doing so, we address new fundamental questions in the field of quantum complexity. Answering them requires using recent techniques and partially answering some new or open questions.

As described in~\cite{bender2002testing}, for bounded-degree directed graphs there are two natural query models: (i) the \emph{unidirectional} model, where the algorithm is allowed to query the outgoing edges of a vertex, but not the incoming edges, and (ii) the \emph{bidirectional} model, where the algorithm can query both the incoming and outgoing edges of a vertex. 
Interestingly, \cite{bender2002testing} showed that strong connectivity is testable in the bidirectional model (i.e., it can be tested with a number of queries that depends on $\eps$ but not on $N$), but it requires $\Omega(\sqrt{N})$ queries in the unidirectional model.
Later, the testability of other graph properties like Eulerianity, vertex and edge connectivity~\cite{orenstein2011testing,yoshida2010testing,forster2020computing,chen2019testability} was also shown in the bidirectional model.
While there is a clear distinction between the two models, Czumaj, Peng and Sohler \cite{czumaj2016relating} showed that if a property is testable in the bidirectional model, then it has a tester with \emph{sublinear} (i.e., $o(N)$) query complexity in the unidirectional model.

In this work, we consider a particularly important problem in the unidirectional model: 
	the problem of testing \emph{subgraph-freeness}. 
	More precisely, we examine the problem of testing ``\emph{$k$-source-subgraph-freeness}'', where the goal is to test $H$-freeness for some constant-sized subgraph $H$ with $k$ ``source components'', where a source component is a strongly connected subgraph that has no incoming edges.
	This problem was first studied by Hellweg and Sohler~\cite{hellweg2013property}, and they provided a testing algorithm that performs $O(N^{1-1/k})$ queries.
They also proved a tight lower bound of $\Omega(N^{2/3})$ for the $k=3$ case (see \cite[Theorem 1 and Theorem 3]{hellweg2013property}).
Very recently, Peng and Wang~\cite{peng2023optimal} proved a matching lower bound for any constant $k$. In particular, they showed that $\Omega(N^{1- \frac{1}{k}})$ queries are necessary for testing $k$-star-freeness (which is a special case of testing $k$-source-subgraph-freeness) in the unidirectional model, for arbitrary $k$ (see \cite[Theorem 1.2]{peng2023optimal}). Notice that asymptotically the complexity of testing $k$-star-freeness becomes $\Omega(N)$.
This also proves that the aforementioned reduction of \cite{czumaj2016relating} can not be made much stronger: for the property of $k$-star-freeness, the separation between the query complexities in the bi- and unidirectional models is maximal, as this property can be tested using constant many queries in the bidirectional model. 

\subsection{Related works on collision finding}\label{introcol} A closely related problem to finding $k$-stars in graphs is finding $k$-collisions in integer sequences. The two mentioned classical papers on subgraph-freeness testing \cite{hellweg2013property, peng2023optimal} actually consider a collision-type intermediate problem for proving their lower bounds. As we are also going to do so, let us look at some related, known results.

The problem of \emph{collision finding} is a ubiquitous problem in the field of algorithm theory with wide applications in cryptography. Here, given a sequence $s$ of $N$ integers, the goal is to find a duplicate in $s$.
If one has the guarantee that $\Theta(N)$ elements of the sequence are duplicated, which is the case for instance when the sequence consists of uniformly random integers from $[N]$, it is well-known that classically $\Theta(\sqrt{N})$ queries are necessary and sufficient due to the birthday paradox.
In the quantum model, this can be solved with query complexity $\Theta(N^{1/3})$ by the algorithm of Brassard, H{\o}yer and Tapp \cite{BHT}.
The matching lower bound was first stated for a specific set of hard instances known as $2$-to-$1$ (i.e., each integer appears exactly twice or not at all) by Aaronson and Shi~\cite{AaronsonShi}. 
For some constant integer $k \ge 3$, those results can be further extended to finding $k$-collisions in a random input with suitable alphabet size, so that it contains $\Theta(N)$ $k$-duplicates with high probability. The classical query complexity for this problem is $\Theta(N^{1-1/k})$ \cite{hellweg2013property,peng2023optimal}, and quantumly it is $\Theta\left(N^{\frac{1}{2}\left(1-\frac{1}{2^k-1}\right)}\right)$ \cite{zhandry_collision}. The situation is more complex for non-random inputs.

Remarkably, the complexity of \emph{testing} $k$-collision-freeness (i.e., the absence of $k$-collisions) is harder to settle on the lower bound side than the \emph{finding} version.
In this work, we are going to focus on the hardness of distinguishing inputs that have linearly many collisions from those that do not have any. 
For $k=2$, the two problems have the same complexity, since intuitively the only way to distinguish is to find a collision. This can be formalized easily in the classical case. Quantumly, this is more challenging, but the lower bound in~\cite{AaronsonShi} proved the hardness of distinguishing between $2$-to-$1$ instances and ones with no duplicate.

However, for larger $k$, distinguishing such inputs might be easier than finding a collision. The classical upper bound of $O(N^{1-1/k})$ queries is straightforward for the finding variant. 
In the lower bounds of \cite{hellweg2013property, peng2023optimal}, the authors consider the distinguishing version, so classically the question is settled. But in the quantum setting, the upper and lower bounds of \cite{zhandry_collision} are tight only for finding $k$-collisions in random inputs, and for the distinguishing variant, we are not currently aware of anything better than the $\Omega(N^{1/3})$ lower bound corresponding to the $k=2$ case. To our knowledge, this problem has not yet been studied in the quantum setting.

\subsection{Our results} \label{sec:intro/our_results}

In this work, we present two lines of results for quantum property testing of graph properties. 

In the first line, we consider the problem of testing $k$-source-subgraph-freeness
in the unidirectional model. This problem is almost maximally hard for large $k$ in the classical regime, and we show that it admits an almost quadratic advantage in the quantum setting.
\begin{theorem}[Restated in \Cref{theo:kstarfree_ub_main}]\label{theo:kstarfree_ub_intro}
The quantum query complexity of testing $k$-source-subgraph-freeness in the unidirectional model is $O\Big(N^{\frac{1}{2}\left(1 - \frac{1}{2^k -1}\right)}\Big)$.
\end{theorem}
In order to prove the above result, we connect it to the problem of finding $k$-collisions.
In \cite{zhandry_collision}, an algorithm is given for finding $k$-collisions in sequences of random integers. We generalize this to the context of graph property testing in two ways: first, finding a subgraph (instead of a collision); and second, considering graphs that are far from being $H$-free (instead of random).

Moreover, we prove that this quantum advantage is nearly tight, by showing a quantum lower bound using the method of dual polynomials.
\begin{theorem}[Corollary of \Cref{lowerbound-col_intro}]\label{theo:kstarfree_lb_intro}
The quantum query complexity of testing $k$-source-subgraph-freeness in the unidirectional model is $\tilde\Omega\big(N^{\frac{1}{2}\left(1-\frac{1}{k}\right)}\big)$.
\end{theorem}

For proving graph property testing lower bounds, both the classical works of \cite{hellweg2013property} and \cite{peng2023optimal} prove collision testing lower bounds using the proportional moments technique of \cite{raskhodnikova2009strong}. 
At the heart of this technique is
a construction of two positive integer random variables, $X_1$ and $X_2$, with different expectations but with the following conditions on the first $k-1$ moments: $\E[X_1]/\E[X_2] = \E[X^2_1]/\E[X^2_2] = \ldots = \E[X^{k-1}_1]/ \E[X^{k-1}_2]$.
Such a construction leads to a randomized query complexity lower bound of $\Omega(N^{1-\frac{1}{k}})$ for various property testing problems such as $k$-collision-freeness \cite{peng2023optimal}.
Having a quantum version of this technique has been identified as an important open problem~\cite{ambainis2016efficient}, since this could be used to pave the way to stronger quantum lower bounds in related settings.
We modestly made progress to this quest for the special case of testing $k$-collision-freeness. 

In \cite{zhandry_collision}, in addition to the algorithm we mentioned, they also prove a matching lower bound showing that their algorithm for \emph{finding} $k$-collisions in random inputs is optimal. However, this time we cannot reuse those techniques for our purpose for two main reasons.
First, the property testing variant of this problem could be easier. Moreover, their lower bound technique requires random inputs and hence it does not apply to our case. This is why we use yet another method, that of dual polynomials, to prove our lower bound.

\begin{theorem}[Restated in \Cref{lowerbound-col}] \label{lowerbound-col_intro}
The quantum query complexity of testing $k$-collision-freeness is $\tilde\Omega\Big(N^{\frac{1}{2}\left( 1 - \frac{1}{k} \right)}\Big)$.    
\end{theorem} 

In the second line of results, we show that not all problems in graph property testing admit such a quantum speedup. This fact even remains valid for the case of undirected graphs. 
For this, we consider the property testing variant of the famous problem of $3$-colorability: namely, distinguishing whether an unknown undirected graph $G$ can be properly colored with $3$ colors, or one needs to modify a large fraction of its edges to make it $3$-colorable. In the classical bounded degree setting, this problem has been studied by \cite{bogdanov2002lower}, who proved a lower bound of $\Omega(N)$ queries. In this work, we present a simple argument that proves that there exists no sublinear quantum tester either for this problem. Our result is stated as follows:  

\begin{theorem}[Restated in \Cref{theo:3colorlbmain}] 
\label{theo:3color_lb_intro}
The quantum query complexity of testing of 3-colorability of undirected bounded-degree graphs is $\Omega(N)$.
\end{theorem}

\subsection{Technical overview}

\subsubsection{Subgraph-finding algorithm} 
We start by describing how to prove the upper bound result of \Cref{theo:kstarfree_ub_intro} for testing $k$-source-subgraph-freeness.
We view the problem as a generalization of the problem of finding $k$-collisions and adapt an existing quantum algorithm for the latter problem.
In \cite{zhandry_collision}, an algorithm is given for finding $k$-collisions in length-$N$ sequences of integers that contain $\Omega(N)$ $k$-collisions (e.g. $k$-to-$1$, or random sequence with appropriate parameters). Their algorithm generalizes the well-known collision finding algorithm of \cite{BHT}. On a high level, the \cite{zhandry_collision} algorithm first finds several $2$-collisions using Grover search like in \cite{BHT}, extends some of them to $3$-collisions in a similar way, and so on until a $k$-collision is found.

On the one hand, instead of random inputs, we consider the problem in the property testing context; and on the other hand, we generalize collision-finding to subgraph-finding.
As a first step let us look at what happens when we consider the property testing version of the $k$-collision problem. In order to be able to use the algorithm of $\cite{zhandry_collision}$, we have to prove that if a length-$N$ sequence is far from $k$-collision-freeness then it contains many $k$-collisions. Notice that the collisions are not necessarily distinct: if the input only contains the same integer $N$ times, it only contains one huge collision, but it is still $\eps$-far from $k$-collision-freeness for any $\eps<1-k/N$. Thus, what we need to show is that there are $\Omega(N)$ many disjoint size-$k$ sets of indices such that for each set, the sequence contains the same character at the positions of the set. This statement is true because otherwise, by modifying all the characters that are in positions contained in a set ($o(N)$ characters in total), we could get a $k$-collision-free sequence which contradicts being far from $k$-collision-freeness.

When we make the second step of turning to testing of subgraph-freeness, we need to prove a variant of this statement: if an $N$-vertex graph $G$ is far from $H$-freeness (for some constant-sized subgraph $H$) then it contains $\Omega(N)$ many ``source-disjoint'' $H$-subgraphs. This means that there are $\Omega(N)$ many such $H$-subgraphs in $G$ that the set of vertices in the source components of each $H$-subgraph are disjoint. We prove this fact in \Cref{claim:eps-far-H} and this allows us to further generalize the approach of \cite{zhandry_collision}: first find several partial solutions where only a few source components of an $H$-subgraph are explored, and gradually extend these (using Grover search coupled with constant-depth BFS) until a complete $H$-subgraph is found.

Notice that this way our algorithm \emph{finds} an $H$-subgraph in $G$ promised that $G$ is far from $H$-freeness. This task is at least as difficult as property testing, where the algorithm only has to \emph{distinguish} whether $G$ is $H$-free or far from any $H$-free graph. So our algorithm provides an upper bound on the property testing variant of $H$-freeness.

\subsubsection{Collision-freeness lower bound} 

Now we will discuss our approach to proving the lower bounds of collision-freeness (\Cref{theo:kstarfree_lb_intro}) and $k$-source-subgraph-freeness (\Cref{lowerbound-col_intro}). We first give a simple reduction from $k$-collision-freeness to $k$-star-freeness, which is a special case of $k$-source-subgraph-freeness. This way, it is enough to prove a lower bound on testing $k$-collision-freeness, and it implies the same result on testing $k$-source-subgraph-freeness.
Since our lower bound approach crucially depends on the (dual) polynomial method, let us start by briefly discussing it. 

\paragraph{The (dual) polynomial method}
The polynomial method is a common way to prove quantum query complexity lower bounds. It relies on the fact that the acceptance probability of a $T$-query bounded-error quantum algorithm is a polynomial of degree at most $2T$ \cite{beals_poly}. This way, for proving a quantum query complexity lower bound on calculating a function $f$, it suffices to argue that any approximating polynomial of $f$ has large degree. 
One of the key properties that such lower bounds exploit is the symmetry that the function $f$ may exhibit, such as invariance under some permutation of the input.
For example, the first tight lower bound of $\Omega(n^{1/3})$ for the collision problem was proved in this way \cite{AaronsonShi}.

The polynomial method can be written in the form of a linear program, of which one can take the dual. By weak LP-duality, when using this dual characterization for proving a lower bound on function $f$, one needs to provide a ``witness'' of the approximating polynomial's high degree, say $\Delta$. This witness is called the \emph{dual polynomial} $\psi$ and, in the easiest case of total Boolean\footnote{We use $\{-1,1\}$ where $-1$ corresponds to the ``true value''.} functions $f:\{-1,1\}^n\to\{-1,1\}$, it needs to have three properties:
\begin{enumerate}
     \item[(i)] High correlation with $f$: $\sum_x f(x)\psi(x)>\delta$;
     \item[(ii)] Normalization: $\sum_x |\psi(x)|=1$;
     \item[(iii)] Pure high degree $\Delta$:  $\sum_{x} p(x)\psi(x)=0$, for every polynomial $p$ with degree $<\Delta$,
\end{enumerate}
where the summations are all over ${x\in\{-1,1\}^n}$.

When the function $f$ is partial, i.e., only defined on a subset $D\subset \{-1,1\}^n$, there is some subtlety that could be handled in two ways (or even in a mixture or both): zero-out the dual polynomial outside $D$ (corresponding to ``unbounded degree''); or rewrite condition (i) accordingly (corresponding to ``bounded degree''):
\begin{enumerate}
     \item[(i')] High correlation with $f$: $\sum_{x\in D} f(x)\psi(x)-\sum_{x\not\in D}|\psi(x)|>\delta$.
\end{enumerate}

\paragraph{Collision function} 
The paper of~\cite{polynomial_strikes_back} also used the dual polynomial method for proving quantum lower bounds for many problems, most of them being open before that work.
Similarly to that paper, we need to take several steps to be able to use the dual polynomial method for the problem of property testing $k$-collision-freeness. 
This problem was not addressed in~\cite{polynomial_strikes_back}.

One of the main conceptual ideas in~\cite{polynomial_strikes_back} is to re-formulate the problem we study as a composition of two simple Boolean functions. In that paper, powerful techniques are also developed in order to design dual polynomials for simple functions that can be composed. A common way of composing dual polynomials (called dual block composition) dates back to~\cite{ShiZ09,Lee09,Sher13}, but~\cite{polynomial_strikes_back} provides new tools for handling it efficiently. We are going to reuse some of them, and also extend one in a way.

The first step is to find the right problem that can fit in the framework.
We introduce a partial symmetric function $F$ defined on input strings $s=(s_1,\dots,s_N)\in[R]^N$.
The domain of $F$ corresponds to the following promise: either $F$ has no $k$-collision, or it has many $k$-collisions occurring for distinct values. More formally, 
$$F(s) =
\begin{cases}
-1 & \text{if no integer occurs at least } k \text{ times in } s, \\
1 & \text{if more than } \gamma R \text{ distinct integers occur at least } k \text{ times in } s, \\
\text{undefined} & \text{otherwise}.
\end{cases}$$
This partial function is not a property testing problem, however it corresponds to a special case of testing $k$-collision-freeness, which is therefore enough to prove lower bounds.

\paragraph{Binary encoding}
Now we encode the input string $s=(s_1,\dots,s_N)\in[R]^N$ into binary variables $x_{i,j}$ storing whether $s_i=j$, as in~\cite{Aaronson02}. Doing so, starting from the function $F$ above, we end up with a function $f$ defined over binary variables satisfying several symmetries, under the permutation of either $i$ or $j$ in $x_{i,j}$. 

Moreover, the symmetries of $f$ allow the extension of the initial function $f$ from the very restricted set of binary inputs correspond to valid strings, to the more general set of  binary inputs of Hamming weight $N$~\cite{ambainis2016efficient}.
With further technicalities 
one can also extend $f$ to all binary inputs of Hamming weight \emph{at most} $N$~\cite{Bun_Thaler_20}.
This is fundamental because instead of being forced to zero out the dual polynomial outside the domain of $f$, we only need to do so on inputs of Hamming weight higher than $N$.
Using the symmetry of $f$, it can be shown that a lower bound on this modified, Boolean version implies a lower bound on the original $k$-collision problem. 

This way we end with two promises on the binary encoding of the input. The first one comes from the function $F$ itself: we have the promise that the input contains either no $k$-collision or it has many ones with different values. 
The second promise is the consequence of the encoding: we want the binary encoding to have Hamming weight at most $N$. Let $D$ denote the set of binary strings satisfying both promises, and let $H_{\le N}$ denote the set of binary strings with Hamming weight at most $N$. 
For this case we use the ``double-promise'' version of the dual polynomial method, where, in order to prove that every $\delta$-approximating polynomial of $f$ has degree at least $\Delta$, the dual polynomial has to satisfy four conditions, where the fourth one corresponds to zeroing out $\psi$ on large Hamming weight inputs~\cite{polynomial_strikes_back}:
\begin{enumerate}
    \item[(i')] High correlation with $f$: $\sum_{x\in D} f(x)\psi(x)-\sum_{x\in H_{\le N}\setminus D}|\psi(x)|>\delta$;
    \item[(ii)] Normalization: $\sum_x |\psi(x)|=1$;
    \item[(iii)] Pure high degree $\Delta$: $\sum_x p(x)\psi(x)=0$, for every polynomial $p$ with degree $<\Delta$; 
    \item[(iv)] No support on inputs with large Hamming weight: $\psi(x)=0$, for every $ x\notin H_{\le N}$. 
\end{enumerate}

\paragraph{Composition}
Coming back now to the definition of our Boolean function $f$, one can rewrite it as a composition of simpler functions: $\gapor_R^\gamma\circ\thr_N^k$, where by composition we mean $(g\circ h)(x)=g(h(x_1),\dots,h(x_n))$ (where $x=(x_1,\dots,x_n)$ and each $x_i$ is a binary vector of appropriate dimension) and the domain is restricted to bit strings of Hamming weight at most $N$. Note that here, $x_j=(x_{1,j},x_{2,j},\ldots,x_{N,j})$.
Above, $\thr_N^k$ is the threshold function: it is $-1$ if the input bitstring contains at least $k$ many $-1$ (true) values, and is $1$ otherwise; and $\gapor_R^\gamma$ is the gap version of $\mathrm{OR}$, which is $1$ if the input is all-1, $-1$ if the input contains at least $\gamma R$ many $-1$ values, and is undefined otherwise.

In order to give a dual polynomial for this composed function, we start from a dual polynomial for each part of the composition ($\phi$ and $\psi$), which were already given in \cite{polynomial_strikes_back} (in different contexts). Then we use a known way~\cite{ShiZ09,Lee09,Sher13} of composing dual polynomials called the \emph{dual block composition}, which provides a nearly good dual polynomial $\phi\star \psi$ for $\gapor_R^\gamma\circ\thr_N^k$. 
Indeed, by construction, the normalization (ii) and the pure high degree (iii) are guaranteed.
However, the issue is that it is not $0$ on bitstrings of large Hamming weight thus (iv) is not satisfied, and and the high correlation (i') still has to be proved.

To fix (iv), we use another result of \cite{polynomial_strikes_back} which provides another dual polynomial $\zeta$, that is close to $\phi\star \psi$ and that is $0$ on inputs having Hamming weight larger than $N$. 
Also, it only changes the pure high degree by a polylogarithmic factor.
Now the only remaining task is to prove a large enough correlation (i') of $\phi\star \psi$ and $\gapor_R^\gamma\circ\thr_N^k$ so that $\zeta$ still has high enough correlation.
This high correlation proof (\Cref{claim:second_point}) is the most technical part of our paper.

\paragraph{High correlation: proof of \Cref{claim:second_point}}
The statement we prove is the following high correlation bound:
$$\sum_{x\in D}(\phi\star\psi)(x)\cdot (\gapor_R^\gamma\circ\thr_N^k)(x)-\sum_{x\notin D}|(\phi\star\psi)(x)|\ge 9/10.$$ 
In the proof of this lemma we use \Cref{pro:like_5.5_5.6} that is a more general statement of some techniques used in several proofs of \cite{polynomial_strikes_back}. But then we need to diverge from their proof because it crucially relies on a certain one-sided error property (in the sense of \cite[Lemma 6.11]{polynomial_strikes_back}) of the inner function of the composition which is OR function in their case. Our inner function is the threshold function which does not satisfy this property, so we have to use some other properties of the dual polynomials in our proof. This different proof technique could be a step towards obtaining a more general lower bound technique.

In fact, the authors in \cite{ambainis2016efficient} stated it as an open question if one could use a variant of the proportional moments technique for proving better quantum lower bounds. 
We leave this question open, and conjecture that a similar result holds in the quantum setting with a lower bound of  $\Omega\Big(N^{\frac{1}{2}\left( 1 - \frac{1}{k} \right)}\Big)$. One can consider this work as a proof of this conjecture for the special case of $k$-collision-freeness, and we hope that it will serve as a step towards proving it in general.

\paragraph{Comparison with \cite{polynomial_strikes_back}}
The structure and several elements of our 
lower bound proof come from \cite{polynomial_strikes_back}.
Using their work is it relatively easy to consider the binary encoding of the $k$-collision function, extend its domain and relate it to the composition of easier functions. They also provide dual polynomial we can use and we can compose them using the dual block composition. We can even use another result of theirs that provides a dual polynomial that is zero on large Hamming weight inputs. Then the only remaining task is to prove the high correlation, for which we need to diverge from this paper because their respective proof relies on a one-sided error property that does not hold for our problem. This way we prove our result in a more difficult, two-sided error setting.

\subsubsection{3-colorability lower bound} 
Let us now discuss our approach to proving the linear lower bound on quantum query complexity of testing $3$-colorability (\Cref{theo:3color_lb_intro}).
Before proceeding to present our approach, let us briefly discuss the classical lower bound of testing $3$-colorability.
To prove the classical lower bound of $3$-colorability, the authors in \cite{bogdanov2002lower} first studied another problem called E$(3, c)$LIN-2, a problem related to deciding the satisfiability of a system of linear equations. More formally, E$(3, c)$LIN-2 considers a system of linear equations modulo $2$, where each equation has $3$ variables and every variable appears in at most $c$ equations. Given such a system of linear equations, the goal is to distinguish if it is satisfiable, or at least some suitable fraction of the equations need to be modified to satisfy it. \cite{bogdanov2002lower} proved that $\Omega(N)$ classical queries to the system of linear equations are necessary for testing E$(3, c)$LIN-2. 

After this, they designed a reduction from 
E$(3, c)$LIN-2 to $3$-colorability such that satisfying instances of E$(3, c)$LIN-2 are reduced to $3$-colorable graphs, and far from satisfiable instances of E$(3, c)$LIN-2 are mapped to far from $3$-colorable graphs. Combining these two arguments, the authors in \cite{bogdanov2002lower} proved that $\Omega(N)$ classical queries are necessary for testing $3$-colorability for bounded degree graphs.

The authors in \cite{bogdanov2002lower} used Yao's minimax method to prove the linear lower bound in testing E$(3, c)$LIN-2. In particular, they designed two distributions $D_{\mathrm{yes}}$ and $D_{\mathrm{no}}$ such that the systems of linear equations in $D_{\mathrm{yes}}$ are satisfiable, whereas the systems of linear equations in $D_{\mathrm{no}}$ are far from being satisfiable. A crucial ingredient of their lower bound proof is a construction of a system of linear equations (represented as a matrix) that are far from being satisfiable, but any $\delta N$ rows of the matrix are linearly independent. Hence, any subset of $\delta N$ entries of the matrix will look uniformly random, and therefore hard to distinguish from a satisfiable instance.

It is a known fact that distinguishing between a uniformly random string and a $\ell$-wise independent string, for an appropriate integer $\ell$, is hard for quantum algorithms (see e.g. \cite{apers2022quantum}). Using this result, we can construct suitable hard instances for E$(3, c)$LIN-2 such that testing E$(3, c)$LIN-2 remains maximally hard (requires $\Omega(N)$ queries) for any quantum algorithm. Combining this hardness result with the reduction from E$(3, c)$LIN-2 to $3$-colorability, we finally prove that $\Omega(N)$ quantum queries are necessary for testing $3$-colorability. We formally prove this in \Cref{sec:quantum_3_coloability_lb}.  

Later, in \cite{yoshida2010query}, the authors used various reductions to $3$-colorability to argue that a number of other important problems including testing Hamiltonian Path/Cycle, approximating Independent Set/Vertex Cover size etc, are maximally hard to test in the classical model. As a corollary of our quantum lower bound, we also obtain maximal quantum query complexity for these problems.

\subsection{Open problems}
Our work raises several important open questions.
First, there is still a gap between our lower and upper bounds on the quantum query complexity of testing $k$-collision-freeness. In \cite{improved_kDist}, the authors keep using the dual polynomial method to improve the lower bound of \cite{polynomial_strikes_back} for the $k$-distinctness problem. They achieve this by using a slightly different dual polynomial for $\thr_N^k$, where they allow more weight on the false positive inputs.
This makes it impossible to prove the high correlation of the dual and the primal functions, so they use a modified block composition. Our technique might be combined with this other approach to improve our lower bound to~$\tilde\Omega(N^{1/2-1/(4k)})$.

As we mentioned before, the authors in \cite{ambainis2016efficient} stated it as an open question if one could use a variant of the proportional moments result of \cite{raskhodnikova2009strong} to prove optimality of quantum property testers in the unidirectional model. This work may be considered as the first attempt to generalize this technique to the quantum setting. 

In \cite{czumaj2016relating}, it was proved that if a graph property can be tested with $O(1)$ queries in the bidirectional model, then it can be tested using $O(N^{1-\Omega(1)})$ queries in the unidirectional model.
It would be very interesting to investigate if it also implies a quantum tester with query complexity, say $O(N^{1/2-\Omega(1)})$.

%% file: arxiv_version_3/preliminaries_arxiv.tex
\section{Preliminaries}\label{sec:prelim}

\subsection{Notations and basic definitions}

Let us denote $[n]=\{1,\dots,n\}$ and $[n]_0=\{0,\dots,n\}$. When dealing with Boolean variables, we will usually use $b\in\{-1,1\}$ instead of $b\in\{0,1\}$. 
We can get to one from the other easily with the mapping $b\mapsto 1-2b$, or its inverse, which means that $-1$ is going to be treated as the ``true'' or ``accepting'' value. The reason for using $\{-1,1\}$ is that when dealing with dual polynomials, it is easier to use this notation.

We denote by $1^n$ the length-$n$ binary vector made only of $1$s, and respectively $-1^n$. 
The Hamming weight $|x|$ of $x\in\{-1,1\}^n$ is then defined as
the number of $-1$s in $x$, that is $|x|=\#\{i\in[n]: x_i=-1\}$.
Let $H_{\leq w}^n=\{x\in\{-1,1\}^n:|x|\leq w\}$ denote the set of length-$n$ binary vectors with Hamming weight at most $w$.
For any $x \in \mathbb{R}$, $\sgn(x)=1$ when $x \geq 0$, and $-1$ otherwise.

For a polynomial $p$, let $\deg(p)$ denote its degree.
The composition $f\circ g:\{-1,1\}^{nm}\to\{-1,1\}$ of two Boolean functions $f:\{-1,1\}^n\to\{-1,1\}$ and $g:\{-1,1\}^m\to\{-1,1\}$ is defined as $(f\circ g)(x)=f(g(x_1),\dots,g(x_n))$ where $x=(x_1,\dots,x_n)$ with each $x_i\in\{-1,1\}^m$.

A directed graph or digraph $G=(V,E)$ is a pair of a vertex set $V$ and an edge set $E$. The latter consists of directed edges that are ordered pairs of vertices: we say that $(u,v)\in E$ is directed from $u$ to $v$ where $u,v\in V$. We say that there is a directed path from $s=v_0$ to $t=v_{l+1}$ (with $s,t\in V$) if there exists an integer $\ell$ and vertices $v_1,\dots,v_\ell\in V$ such that $\forall i\in[\ell]_0:\ (v_i,v_{i+1})\in E$. A digraph $G=(V,E)$ is called \emph{strongly connected} if for every $u\in V$ and $v\in V\setminus\{u\}$, there exists a directed path from $u$ to $v$. A \emph{subgraph} of a graph $G = (V, E)$ is any graph $G' = (V', E')$ satisfying $V' \subseteq V$, $E' \subseteq E$ and $E'\subseteq V'\times V'$.

Finally, throughout this work, notations $O(\cdot)$ and $\Omega(\cdot)$ will be hiding the dependencies on parameters $\eps$, $k$ and $d$ that we consider to be constants. Additionally, we will use $\widetilde{O}(\cdot)$ and $\widetilde{\Omega}(\cdot)$, where we hide poly-logarithmic dependencies on the parameters.

\subsection{Query complexity}
In query complexity, we consider inputs $x\in \Sigma^I$ over a finite alphabet $\Sigma$ and indexed by a set $I$.
They are not given explicitly to the algorithm. Instead, the algorithm has query access to an input oracle $\cO_x:I\to\Sigma$ encoding $x$ by $\cO_x(i)=x_i$. 
Quantumly, the query access is described by the unitary operator $O_x\ket{i}\ket{z}= \ket{i}\ket{z \oplus x_i}$, for $z\in \Sigma$ and $i\in I$, where $\oplus$ is usually the bit-wise exclusive-OR operation up to some binary encoding of the elements of $\Sigma$. But our lower bound technique applies to any reversible operation $\oplus$.

Query complexity measures the minimum number of queries that an algorithm has to make in order to decide whether a property $\cP:D_\cP\subseteq \Sigma^I\to\{-1,1\}$ is satisfied, for an arbitrary input $x \in D_\cP$. Since the work of \cite{beals_poly}, it has been known that in the Boolean case (i.e., when $\Sigma=\{-1,1\}$), the acceptance probability of a $T$-query bounded-error quantum algorithm is a multivariate polynomial (in all the $x_i$'s) of degree at most $2T$.

In this work, two kinds of inputs are going to play a crucial role.
    When the input is a sequence $s=(s_1,\dots,s_N)$ of positive integers $\le R$, then $I=[N]$ and $\Sigma=[R]$.
    
    In the undirected bounded-degree graph model, we have query access to the adjacency list of an undirected graph $G=(V,E)$ with maximum degree $d$, represented as an oracle $\cO_G: V \times [d] \rightarrow V \cup \{\bot\}$. Then we can set $I=V\times [d]$ and $\Sigma=V\cup\{\bot\}$,  such that for any $v \in V$ and $i \in [d]$, we have the following:
    $$\cO_G(v,i)=
\begin{cases}
    w , & \text{if $w\in V$ is the $i$-th neighbor of $v$};\\
    \bot, & \text{if deg}(v) <i.
\end{cases}
$$    

For bounded-degree directed graphs, there exist two query models.
In the \emph{bidirectional model}, we have access to both the outgoing and incoming edges of each vertex.
Correspondingly, it is imposed that both the in- and out-degrees of a vertex are bounded by $d$.
In the \emph{unidirectional model}, we can only make queries to the adjacency list of the outgoing edges, and we impose only that the out-degrees of a vertex are bounded by $d$. Since in this work the primary focus will be on the latter model, let us formally define it below.

In the unidirectional bounded-degree graph model, we have query access to the adjacency list of a digraph $G=(V,E)$ where the out-degree of every vertex is at most $d_{\mathrm{out}}$: for all $v\in V$: $\text{deg}_{\mathrm{out}}(v)\le d_{\mathrm{out}}$. This access is represented as an oracle $\cO^{\mathrm{out}}_G: V \times [d_{\mathrm{out}}] \rightarrow V \cup \{\bot\}$. Then we can set $I=V\times [d_{\mathrm{out}}]$ and $\Sigma=V\cup\{\bot\}$,  such that for any $v \in V$ and $i \in [d_{\mathrm{out}}]$, we have the following:
    $$\cO^{\mathrm{out}}_G(v,i)=
\begin{cases}
    w , & \text{if $w\in V$ is the $i$-th out-neighbor of $v$};\\
    \bot, & \text{deg}_{\mathrm{out}}(v) <i.
\end{cases}
$$

For completeness, we note that in some of the previous works on the unidirectional model, the authors do impose the degree bound on both the out- and in-degree (see, e.g., \cite{czumaj2016relating}).
This is mostly because this makes for an easier comparison between the uni- and bidirectional models, as this way they allow the same set of graphs.
In this work, we assume that only the out-degrees of the vertices are bounded by $d$.

\subsection{Property testing}
In total decision problems, the algorithm has to decide if the input satisfies a property or not. In the case of property testing, the question is relaxed: the algorithm has to distinguish inputs that satisfy the property from those that are ``far'' (according to some distance measure) from any input that satisfies it.

The choice of distance measure usually depends on the query model considered. As discussed before, the general query access can be viewed as black box access to the input $x\in \Sigma^I$ where querying an index $i\in I$ reveals $x_i\in\Sigma$. This way, the distance of two objects is described
as the proportion of positions where they differ:
$$\text{$x$ is $\eps$-far from $y$} \iff |\{i\in I: x_i\neq y_i\}| \geq \eps |I|.$$

Applying this to the case of bounded-degree graphs with degree bound $d$ and query access to the adjacency list, the distance of two graphs is the number of edges where they differ divided by $|V|d$. The distance of an object $x$ from a property $\cP$ is the minimum distance between $x$ and any object that satisfies $\cP$.

\begin{definition}[Property Testing] \label{def:proptest}
    Let $0<\eps<1$ be a constant. An algorithm $\mathcal{A}$ is an \emph{$\eps$-tester} for the property $\cP$ if
    \begin{enumerate}
        \item For all $x \in \cP$:  $\Pr[\mathcal{A}(x)=\textnormal{accept}]\ge 2/3$;
        \item For all $x$ that are $\eps$-far from $\cP$: $\Pr[\mathcal{A}(x)=\textnormal{accept}]\le1/3$.
    \end{enumerate}
\end{definition}

Notice that no restriction is given on the acceptance probability of the algorithm for inputs that do not satisfy $\cP$ but are $\eps$-close to it.

\subsection{Grover search} 
We are not going into details about quantum computing, because for this paper, it suffices to state very few results in the field and use them in a black box way.
One of the most important results in quantum computing is Grover's algorithm for unordered search \cite{Grover}: finding a marked element in an unordered database of size $N$ takes $\Theta(N)$ queries classically, but the quantum query complexity of the same task is $\Theta(\sqrt{N})$.
We are going to use a particular variant of this result that has been used many times in the literature (see e.g., \cite[Item 3 in Section 2.2]{Ambainis_search}, which was implicitly proved in \cite{tight_lb_grover}).
\begin{theorem} \label{th:Grover}
Let $1\leq t_0\leq N$.
There exists a quantum algorithm that, given $t_0$ and query access to any function
 $f:S\to\{0,1\}$, makes $O(\sqrt{N/t_0})$ queries to $f$ and outputs
either ``not found'' or an element uniformly at random in $f^{-1}(1)$.
Moreover, when $|f^{-1}(1)|\geq t_0$ the later occurs with high constant probability.
\end{theorem}
\begin{remark} \label{rem:Grover}
    In practice, we will use this theorem when querying $f(z)$, for $z\in S$, requires making $c$ queries to an input graph $G$ with vertex set $S$. In that case the total query complexity to $G$ is $O(c\sqrt{N/t_0})$.
\end{remark}

\subsection{Problem definitions}
We now formally define the problems we study and argue about certain relations between them.
While the problems are phrased as decision problems, ultimately we will care about the quantum query complexity for testing the corresponding properties.
The complexity is going to be parameterized by a parameter $k$. Moreover, the parameters $k, \eps$ and the degree bound $d$ are all considered to be constants throughout this paper.

Let us start with some definitions that will be useful to define our problems precisely.

\begin{definition}[Source component]
Let $H=(V,E)$ be a digraph. A set $S\subseteq V$ is called a \emph{source component} if it induces a strongly connected subgraph in $H$, and there is no edge from $V\setminus S$ to $S$ in $H$.
\end{definition}

\begin{definition}[$k$-star]
A $k$-star is a digraph on $k+1$ vertices and $k$ edges with one center vertex, and $k$ source vertices connected to the center vertex.
\end{definition}
Notice that a $k$-star has $k$ source components each consisting of a single vertex.

We will now state the decision variants of several problems.
The ``property'' corresponding to a decision problem is the set of inputs that should be accepted in the decision problem.

\paragraph*{$k$-Source-Subgraph-Freeness}
\begin{description}
\item[Parameter:] Graph $H$ of constant size with at most $k$ source components
\item[Query access:] $d$-bounded out-degree directed graph $G$ on $N$ vertices (unidirectional model)
\item[Task:] Accept iff $G$ is $H$-free, that is, no subgraph of $G$ is isomorphic to $H$ 
\end{description}

In \cite{hellweg2013property, peng2023optimal}, the authors examine the classical query complexity of testing $k$-source-subgraph-freeness.
They consider the bounded-degree unidirectional model, albeit with a bound on both the in- and out-degrees.

For proving a lower bound, we will look at a special case of the main problem: $k$-star-freeness. Notice that since a $k$-star has $k$ source components, a lower bound for this problem implies the same lower bound for the more general $k$-source-subgraph-freeness problem. 

\paragraph*{$k$-Star-Freeness}
\begin{description}
\item[Parameter:] Integer $k\geq 2$
\item[Query access:] $d$-bounded out-degree directed graph $G$ on $N$ vertices (unidirectional model)
\item[Task:] Accept iff $G$ is $k$-star-free, that is, no subgraph of $G$ is isomorphic to the $k$-star
\end{description}

For the lower bound on $k$-star-freeness testing, we are going to use as a ``helper problem'' the testing variant of the $k$-collision problem.

\paragraph*{$k$-Collision-Freeness}
\begin{description}
\item[Parameter:] Integer $k\geq 2$
\item[Query access:] Sequence of integers $s=(s_1,\dots,s_N)\in[R]^N$
\item[Task:] Accept iff $s$ is $k$-collision-free, i.e., there is no $i_1,\dots,i_k\in [N]$ with $s_{i_1}=\dots=s_{i_k}$
\end{description}

As discussed in the introduction (\Cref{introcol}), very little was known about the property testing version of this problem prior to this work.
We only know the that the complexity is $\Theta(N^{1/3})$ when $k=2$, and it is between $\Omega(N^{1/3})$ and $O\left(N^{\frac{1}{2}\left(1-\frac{1}{2^k-1}\right)}\right)$ for larger $k$.

\paragraph*{Reduction from $k$-collision-freeness to $k$-star-freeness}
Now we are going to prove that testing $k$-collision-freeness can be reduced to testing $k$-star-freeness (or more generally to testing $k$-source-subgraph-freeness). Thus, a lower bound on testing $k$-collision-freeness yields a lower bound on testing $k$-source-subgraph-freeness. Also, an algorithm for testing $k$-source-subgraph-freeness yields an upper bound on testing $k$-collision-freeness.

While the proof goes similarly to \cite[Theorem 3]{hellweg2013property}, our reduction is not identical
because we have a slightly different ``helper problem''. Since they consider that the in-degree of vertices to be bounded as well, for the collision problem, they assume that the sequence does not contain any collision of size larger than $k$ (defined as \emph{$k$-occurrence-freeness}).

\begin{proposition} \label{pro:coll_star_reduction}
    The problem of $\eps$-testing $k$-collision-freeness of a sequence from $[R]^N$ can be reduced to $\frac{\eps N}{d(N+R)}$-testing $k$-star-freeness of an $(N+R)$-vertex sparse directed graph with out-degree bound $d\ge 1$.
\end{proposition}

\begin{proof}
    Let us assume that we have an algorithm that solves the $k$-star-freeness testing problem on graphs with out-degree bound $d\ge 1$, and we want to use it to test $k$-collision-freeness of a sequence $s=(s_1,\dots,s_N)\in[R]^N$. We construct a digraph $G$ that has $N$ outer vertices $u_1,\dots,u_N$ and $R$ inner vertices $v_1,\dots,v_R$; edges only exist from the outer vertices towards the inner ones such that $u_i$ is connected to $v_j$ iff $s_i=j$. Observe that the maximum out-degree in $G$ is 1, so its out-degree is bounded by $d$ for any $d\ge 1$.

    It is clear that $s$ is $k$-collision-free iff $G$ is $k$-star-free. On the other hand, if $s$ is $\eps$-far from $k$-collision-freeness, it implies that more than $\eps N$ edges have to be deleted in $G$ to make it $k$-star-free. Thus $G$ is $\eps'=\frac{\eps N}{d(N+R)}$-far from $k$-star-freeness.
\end{proof}

%% file: arxiv_version_3/algorithm_arxiv.tex
\section{Quantum algorithm for testing subgraph-freeness}\label{sec:quantum_algo}

In this section, we prove that there is a quantum speedup for testing $H$-freeness in directed graphs with $d$-bounded out-degree, for any graph $H$ that has $k$ source components.
For large but constant $k$, the speedup is near-quadratic. This problem was studied in \cite{goldreich2002property} in the classical setting.
Our algorithm can be seen as a generalization of the one in \cite{zhandry_collision} to graphs. 
Let us start with the definition of source-disjointness which will be used in the analysis of our algorithm.

 \begin{definition}[Source-disjointness]
    Let $G$ be a directed graph such that it contains two subgraphs $H_1$ and $H_2$. We say that $H_1$ and $H_2$ are \emph{source-disjoint} if the union of the source components of $H_1$ is disjoint from the union of the source components of $H_2$.
\end{definition}

Moreover, we need to prove the following simple proposition.
It shows that if $G$ is far from being $H$-free, then it contains 
many source-disjoint copies of $H$, that is, copies of $H$ that are source-disjoint subgraphs of $G$.

\begin{proposition} \label{claim:eps-far-H}
	Let $H$ be an $h$-vertex graph with $k$ source components.
    Assume that a $d$-bounded out-degree directed graph $G$ on $N$ vertices is $\eps$-far from $H$-free.
	Then $G$ contains at least $\eps N/h=\Omega(N)$ source-disjoint copies of $H$.
\end{proposition}

\begin{proof}
We prove the result by contraposition.   
Consider  a maximal set $M$ of source-disjoint copies of $H$ in $G$, and assume that $|M|\leq \eps N/h$.
Let $U$ denote the union of all the vertices in the source components of the copies in $M$. This implies that if one deletes all the outgoing edges of all the vertices in $U$, then $G$ becomes $H$-free. Indeed, if there remained an $H$-copy then all its source components are disjoint from $M'$ (as in a source component every vertex has at least one outgoing edge), contradicting the fact that $M$ was maximal.

Since $|U|\leq |M|\cdot h $, the number of those deleted edges is at most $|U|\cdot d \le |M|\cdot hd\le \eps N d$.
Therefore, the resulting graph is both $H$-free and $\eps$-close to the original graph $G$. 
This proves the contraposition of the proposition.
\end{proof}

We will use Breadth-first search (BFS) 
in order to explore $G$ layer by layer: starting from a given vertex, first it explores its direct (out-)neighbors, then their unexplored (out-)neighbors etc.
We will run BFS up to some limited depth $\ell$, so that
the depth-$\ell$ BFS algorithm has query complexity at most $d^\ell$ where $d$ is the maximum (out-)degree of $G$. In our application, $d$ and $\ell$ are constants, and thus the query complexity is $d^\ell=O(1)$.

\subsection{The algorithm for $k=2$}
To illustrate our algorithm, we first consider the $k=2$ case to build some intuition.
Here, our algorithm generalizes the BHT algorithm for collision finding \cite{BHT} in the context of graphs.
The high level idea is that if we manage to sample a vertex from each of the two source components of an $H$-subgraph (a collision), then by querying their ``surroundings'' we will discover the $H$-instance.
In the following, we set $h=|V(H)|$, the number of vertices in $H$.
\begin{enumerate}
	\item Sample a uniformly random vertex subset $\cS$ of size $t=\Theta(N^{1/3})$ in $G$. Perform a depth-$h$ BFS from every vertex in $\cS$. 
	\item Perform Grover search over the remaining vertices $V \backslash \cS$ in the following way. 
	A vertex $v$ is marked if there exists another vertex $u\in \cS$ such that $u$ and $v$ are from the 2 different source components of an $H$ subgraph of $G$. 
	\item If any occurrence of $H$ in $G$ is found,
	output \textbf{Reject}. Otherwise, output \textbf{Accept}.
\end{enumerate}

Note that if $G$ is $H$-free, then the above algorithm will always accept. Now we need to argue that if $G$ is $\eps$-far from being $H$-free, then with constant probability, the above algorithm will find a copy of $H$ and thus will output reject.

By \Cref{claim:eps-far-H}, with high probability, a constant fraction of the $t$ vertices in $\cS$ are part of a source component in source-disjoint $H$-subgraphs of $G$.
For such vertices, the BFS in step 1~will discover the entire source component, as well as all other vertices reachable from that source component in $H$.
Then, in step 2~we search for a vertex that is in the remaining source component of such an instance of $H$ that we already partly discovered. This can be verified by doing a depth-$h$ BFS from it and checking if this completes an $H$-instance with one of the previously sampled vertices' neighborhoods.
As we mentioned, by \Cref{claim:eps-far-H} with high probability there are $\Omega(t)$ many marked vertices.
This proves the correctness of the algorithm.

Finally, we bound the algorithm's query complexity.
Step 1~makes $O(t) = O(N^{1/3})$ many (classical) queries.
In step 2 we use \Cref{th:Grover} and \Cref{rem:Grover}: checking whether a vertex is marked requires running a depth-$h$ BFS from it, which costs $c=O(1)$ queries. We argued that there are $\Omega(t)$ many marked vertices, so Grover search makes $O(\sqrt{N/t}) = O(N^{1/3})$ quantum queries.

\subsection{The algorithm for general $k$}

We are now ready to state our general upper bound result.
The algorithm and proof follow the same lines as the $k=2$ case.

\begin{theorem}[Restatement of \Cref{theo:kstarfree_ub_intro}]\label{theo:kstarfree_ub_main}
	Let $H$ be a digraph of constant size with $k$ source components.
	The quantum query complexity of testing $H$-freeness of an $N$-vertex graph with bounded out-degree in the unidirectional model is $O\Big(N^{\frac{1}{2}\left(1 - \frac{1}{2^k -1}\right)}\Big)$.
\end{theorem}

\begin{proof}
In order to extend the $k=2$ case described above to larger $k$, we first try to find many partial $H$-instances with $k-1$ source components found, and then extend one of them to a complete $H$-instance.  We present a brief description of our algorithm below, where $h$ is the number of vertices of $H$:
	\begin{enumerate}
		\item Sample a uniformly random vertex subset $\cS_1$ in $G$ of size $t_1$.
		Perform a depth-$h$ BFS from every vertex in $\cS_1$.
		Let $\cS_1'=\cS_1$.
		\item 
		For iterations $i =2$ to $k-1$, do the following:
		\begin{enumerate}
			\item  Perform a Grover search $t_i$ times on the vertices $V \setminus \cS_{i-1}'$ in the following way. 
			A vertex $v$ is marked if there exist $i-1$ other vertices $u_j\in \cS_{j}$ for each $j\in[i-1]$ such that $u_1,\dots,u_{i-1}$ and $v$ are from $i$ different source components of an $H$ subgraph of $G$. If we do not find $t_i$ vertices like this, output \textbf{Reject}, otherwise let $\cS_i$ denote the set of the vertices $v$ that we found.
			\item Set $\cS_{i}'= \cS_{i-1}' \cup \cS_i $.
		\end{enumerate} 
		\item Perform Grover search on $V \backslash \cS_{k-1}'$ to find a complete $H$-instance. I.e., a vertex $v$ is marked if there exist $k-1$ other vertices $u_j\in \cS_{j}$ for each $j\in[k-1]$ such that $u_1,\dots,u_{k-1}$ and $v$ are from the $k$ different source components of an $H$ subgraph of $G$.
		\item If any occurrence of $H$ in $G$ is found,
		output \textbf{Reject} and terminate the algorithm. Otherwise, output \textbf{Accept}.
	\end{enumerate}
	The correctness proof is similar to the $k=2$ case. \Cref{claim:eps-far-H} tells us that in $\cS_1$ there are $\Omega(t_1)$ many vertices that are from a source component of an $H$-copy. Because of the source-disjointness of the $H$-copies, when $i=2$, there are $\Omega(t_1)$ many 1-partial solutions that can be extended to a complete $H$ instance by disjoint remaining source components. As Grover search provides uniformly random marked elements, a constant fraction of the $t_2$ many 2-partial solutions are actually extendable to $H$ in a similar, disjoint way. This continues to be true in each iteration: (with high probability) a constant fraction of the $t_{i-1}$ many $(i-1)$-partial solutions are extendable to complete $H$ instances by disjoint remaining source components. This way, the last step is going to find an $H$-subgraph with high probability.
	
	To bound the query complexity, first note that in every application of Grover search, checking whether a vertex is marked (depth-$h$ BFS) takes $O(1)$ queries. The first iteration's Grover searches find $t_2$ partial $H$-instances with $2$ of its source components found, which takes $O(t_2\sqrt{N/t_1})$ queries (by \Cref{th:Grover}). Similarly, for $i$-th iteration there are $\Omega(t_{i-1})$ marked elements (see the previous argument), so the algorithm performs $O(t_i \sqrt{N/t_{i-1}})$ quantum queries for every $i \in [k-1]$. Finally, finding one complete $H$-instance costs $O(\sqrt{N/t_{k-1}})$ queries. Thus the total query complexity is $O(t_1+\sum_{i=1}^{k-1}t_{i+1}\sqrt{N/t_i})$ with $t_k=1$.
	Similar to the multi-collision algorithm in \cite[Section~3]{zhandry_collision}, we can equate all terms by setting $t_i = \Theta\Big(N^{\frac{2^{k-i}-1}{2^k-1}}\Big)$, which yields the final quantum query complexity $O\Big(N^{\frac{1}{2}\left(1-\frac{1}{2^k-1}\right)}\Big)$.
	
	We note
	that there is no need for a $\textnormal{polylog}(N)$ factor in the query complexity, which could come from a commonly used way to boost up the success probability of Grover's algorithm. This stems from two observations. 
	First, consider the case where $K$ among $N$ elements are marked, with a given lower bound $L \leq K$, and we wish to find $R \leq L$ such elements. If $R \ll L$, then (say) $100 R$ repetitions of Grover should return at least $R$ marked elements with probability at least $2/3$ while making $O(R  \sqrt{N/L})$ queries, without extra log-factors. This is due to the fact that one can simply ignore any unsuccessful Grover runs. In our case we set $R=t_{i+1} \ll L=t_i$. Finally, since there are $k$ iterations in the algorithm and $k$ is constant, a factor of $\log k$ would not add up to the query complexity of our algorithm in terms of $N$.
\end{proof}

%% file: arxiv_version_3/kstarfreeness_test_arxiv.tex
\section{Collision-freeness lower bound}\label{sec:quantum_lb}

As discussed in \Cref{sec:prelim}, we are going to prove a lower bound on the problem of testing $k$-collision-freeness.

\begin{theorem}[Restatement of \Cref{lowerbound-col_intro}]\label{lowerbound-col}
Let $ k\geq 3$ and $0<\eps<1/(4^{k-1}\lceil 20 (2k)^{k/2}\rceil)$ be constants and $N$ 
be a large enough positive integer.
Then the quantum query complexity of the $\eps$-testing of $k$-collision-freeness 
of a sequence of integers $S=(s_1,\dots,s_N)\in[N]^N$ with parameter $\eps$ 
is $\Omega(N^{1/2 - 1/(2k)}/\ln^2 N)$.
\end{theorem}

The proof of the theorem is at the end of \Cref{proof:th_main_lb}.
Observe that \Cref{theo:kstarfree_lb_intro} is implied by \Cref{lowerbound-col} and the reduction in \Cref{pro:coll_star_reduction}.
Our proof mostly follows the structure of \cite[Section 6.1]{polynomial_strikes_back}, 
and in particular it uses the notion of dual polynomial for non-Boolean partial symmetric functions.
Our main technical contribution in this section is the proof of \Cref{claim:second_point}, because the corresponding proof in \cite{polynomial_strikes_back} crucially relies on a fact that does not hold for our problem. We will discuss it in detail below.

In the following, we first state some general results related to the polynomial method for non-Boolean functions, then we use these results for our problem to state the exact statement that we prove in the technical part.

\subsection{The (dual) polynomial method} \label{sec:poly_boolean}
\paragraph*{For Boolean functions}
We consider a property on Boolean vectors as a function $f:D\subseteq\{-1,1\}^n\to\{-1,1\}$. Since the work of \cite{beals_poly}, it has been known that the acceptance probability $\mathsf{acc}(x)$ of a $T$-query bounded-error quantum algorithm on input $x\in D$ is a polynomial of degree at most $2T$. Thus, the polynomial $p(x)=1-2 \cdot \mathsf{acc}(x)$ must be a good approximation of $f$. 
Observe that $p(x)$ remains bounded outside $D$ since $\mathsf{acc}(x)$ remains a probability defined by the algorithm, with no constraint.

In order to formalize this, we first define the notion of approximate degree of a Boolean function, and then relate it to its query complexity.
\begin{definition}[Approximate bounded degree]
    Let  $f:D\subseteq\{-1,1\}^n \to\{-1,1\}$ and $\delta>0$. 
A polynomial $p:\{-1,1\}^n\to\mathbb{R}$     \emph{$\delta$-approximates $f$ on $D$} if
 $$\forall x\in D:\ |f(x)-p(x)|<\delta \quad\text{ and }\quad\forall x\in\{-1,1\}^n\setminus D:\ |p(x)| < 1+\delta.$$
Moreover, the \emph{$\delta$-approximate bounded degree $\bdeg_\delta(f)$ of $f$ on $D$} is the smallest degree of such a polynomial.
\end{definition}

The following lemma connects the quantum query complexity and approximate bounded degree.

\begin{lemma}[\cite{beals_poly,AaronsonAIKS16}]
    Let $f:D\subseteq\{-1,1\}^n\to\{-1,1\}$ and $\delta>0$. 
If a quantum algorithm computes $f$ on $D$ with error $\delta$ using $T$ queries, then there is a
polynomial $p$ of degree at most $2T$ that $2\delta$-approximates $f$ on $D$.    
\end{lemma}
In particular, this implies that the quantum query complexity for computing $f$ with error $\delta$ is $\bdeg_{2\delta}(f)/2$, and so we will focus on proving lower bounds on the approximate bounded degree.

We now turn to a dual characterization of this polynomial approximation.
This method of dual polynomials dates back to~\cite{Sherstov11,ShiZ09} for initially studying communication complexity. Below we refer to some results stated in~\cite{polynomial_strikes_back} for studying query complexity.
\begin{definition}[Pure high degree]
    A function $\psi:\{-1,1\}^n\to\mathbb{R}$ has pure high degree at least $\Delta$ if for every polynomial $p:\{-1,1\}^n\to\mathbb{R}$ with $\deg(p)<\Delta$ it satisfies $\sum_{x\in\{-1,1\}^n}p(x)\psi(x)=0$.
    We denote this as $\phd(\psi)\ge \Delta$.
\end{definition}
One can observe that $\phd(\psi)\ge \Delta$ is equivalent to the fact that all the monomials of $\psi$ are of degree at least $\Delta$. Then by weak LP duality we get the following result.
\begin{theorem}{\cite[Proposition 2.3]{polynomial_strikes_back}}
    Let $f:D\subseteq \{-1,1\}^n\to\{-1,1\}$ and $\delta>0$. Then $\bdeg_\delta(f)\ge \Delta$ iff there exists a function $\psi:\{-1,1\}^n\to\mathbb{R}$ such that
    \begin{align}
        \sum_{x\in D}\psi(x)f(x)-\sum_{x\in \{-1,1\}^n\setminus D}|\psi(x)|>\delta ;\\
        \|\psi\|_1=\sum_{x\in\{-1,1\}^n}|\psi(x)|=1;\\
        \phd(\psi)\ge \Delta.
    \end{align}
\end{theorem}

Now we are going to discuss how to extend these results to non-boolean functions, which is the interesting case for us.

\paragraph*{For non-Boolean partial symmetric functions} \label{sec:dual_poly_non_boolean}
We now consider a property of a sequence of integers as a function $F:D\subseteq [R]_0^N\to\{-1,1\}$. The symbol $0$ will play a special role that will be exhibited later on. Unfortunately one cannot just take the polynomial of those integers. 
The standard approach (see~\cite{Aaronson02}) is to encode $s=(s_1,\dots,s_N)\in[R]_0^N$ into binary variables $x=(x_{i,j})_{i\in[N],j\in[R]_0}\in\{-1,1\}^{N(R+1)}$ encoding whether $s_i=j$ as follows:
$x_{i,j}=-1$ if $s_i=j$, and $x_{i,j}=1$ otherwise. 
Let $H_b^{N(R+1)}\subseteq\{-1,1\}^{N(R+1)}$ be the set of all possible encodings of vectors $s$, that is, for every $i\in[N]$ there is exactly one $j\in[R]_0$ such that $x_{i,j}=-1$.

This way we can represent $F$ as a function $F_b:D_b\to\{-1,1\}$ where $D_b\subseteq H_b^{N(R+1)}$
is the set of valid encodings of $D$. More precisely, each $x\in D_{b}$ satisfies two constraints: (1) $x\in H_b^{N(R+1)}$; and (2) $x$ encodes some $s\in D$.
Since only inputs $x\in H_b^{N(R+1)}$ correspond to possible input sequences of an algorithm, the polynomials derived from a quantum query algorithm might not be bounded outside of that set.
This implies a slight modification on the definition of approximate degree, in order to relate it to query complexity as in~\cite{Aaronson02}.

But before doing this, we are going to relax the constraints on the domain $D_b$ in the case of symmetric functions, while we decrease its dimension.
When $F$ is \emph{symmetric} (i.e., $F(s)=F(s\circ\pi_N)$ for any permutation $\pi_N$ of $[N]$), one can instead define a function $F_{\leq N}$ with weaker constraints by removing the variables corresponding to the symbol $0$.
Define $H_{\leq N}^{NR}$ as the set of length-$(NR)$ binary vectors with Hamming weight at most $N$.
Given any $x\in H_{\leq N}^{NR}$, we define its frequency vector $z(x)=(z_0,z_1,\dots,z_R)$
with $z_j=\#\{i:x_{ij}=-1\}$, for $1\leq j\leq R$, and $z_0=N-z_1-\ldots-z_R$.
From the vector $z(x)$, one can define a valid sequence of integers $s(x)\in[R]_0^N$: it can be any sequence from $[R]_0^N$ that has frequency vector $z(x)$.
Now we can define  $F_{\leq N}$ on domain $D_{\leq N}$ as
$$D_{{\leq N}}=\{x\in H_{\leq N}^{NR} : s(x)\in D\}
\quad\text{ and }\quad F_{\leq N}(x)=F(s(x)).$$
In fact, for the special case of total symmetric functions $F$,
we can transform $F_b$ on $H_b^{N(R+1)}$ to $F_{\leq N}$ on $H_{\leq N}^{NR}$ due to the symmetry of $F$.

In \cite{ambainis_poly} it was proved implicitly that for symmetric $F$, both $F_b$ and $F_{\leq N}$ variants are equally hard to approximate by polynomials. We now define the appropriate notion of approximate degree for $F_{\leq N}$ and relate it to the query complexity of $F$ as in \cite[Theorem 6.5]{polynomial_strikes_back}.
\begin{definition}[Double-promise approximate degree]
    Let $F: D\subseteq [R]_0^N \to\{-1,1\}$ be symmetric and $\delta>0$. 
    Define $H_{\leq N}^{NR}\subseteq\{-1,1\}^{NR}$ and  $F_{\leq N}:D_{\leq N}\subseteq H_{\leq N}^{NR} \to\{-1,1\}$ as above.
    A polynomial  $p:\{-1,1\}^{NR}\to\mathbb{R}$  \emph{double-promise $\delta$-approximates} $F$ on $D$ if $$\forall x\in D_{\leq N}:\ |F_{\leq N}(x)-p(x)|<\delta \quad\text{ and }\quad\forall x\in H_{\leq N}^{NR}\setminus D_{{\leq N}}:\ |p(x)|<1+\delta.$$ 
Moreover, the \emph{double-promise $\delta$-approximate degree $\dpdeg_\delta(F_{\leq N})$ of $F_{\leq N}$ on $D_{\leq N}$} is the smallest degree of such a polynomial.
\end{definition}

The following lemma connects the quantum query complexity and double-promise approximate degree.

\begin{lemma}[\cite{Aaronson02,ambainis_poly},{\cite[Theorem 3.9]{Bun_Thaler_20}}] \label{lem:QC_deg}
   Let $F: D\subseteq [R]_0^N \to\{-1,1\}$ be symmetric and $\delta>0$. 
    Define $H_{\leq N}^{NR}\subseteq\{-1,1\}^{NR}$ and  $F_{\leq N}:D_{\leq N}\subseteq H_{\leq N}^{NR} \to\{-1,1\}$ as above.
If a quantum algorithm computes $F$ on $D$ with error $\delta$ using $T$ queries, then 
there is a polynomial $p$ of degree at most $2T$ that double-promise $2\delta$-approximates $F_{\leq N}$ on $D_{\leq N}$.    
\end{lemma}

As for the Boolean case, this implies that a quantum algorithm computing $F$ with error $\delta$ must make at least $\dpdeg_{2\delta}(F_{\leq N})/2$ queries.
We can now also take the dual of this characterization.
\begin{theorem}[{\cite[Proposition 6.6]{polynomial_strikes_back}}] \label{th:dpdeg0}
    Let  $F:D\subseteq [R]_0^N\to\{-1,1\}$ be . 
    Define $F_{\leq N}:D_{\leq N}\to\{-1,1\}$ as above.
    Then $\dpdeg_{\delta}(F_{\le N})\ge \Delta$ iff there exists a function $\psi:\{-1,1\}^{NR}\to\mathbb{R}$ such that
\begin{align}
    \forall x\in\{-1,1\}^{NR}\setminus H_{\le N}^{NR},\quad \psi(x)=0;\\
    \sum_{x\in D_{{\le N}}}\psi(x)F^{\le N}(x)-\sum_{x\in H_{\le N}^{NR}\setminus D_{{\le N}}}|\psi(x)|>\delta;\\
     \|\psi\|_1=1  \quad \text{ and } \quad
      \phd(\psi)\ge \Delta.
\end{align}
\end{theorem}

\subsection{Preparation} 

Technically, the problem we use in the proof of \Cref{lowerbound-col}
 is slightly more restricted than $k$-collision-freeness: we want to distinguish no $k$-collision from many distinct collisions of size at least~$k$.

\begin{definition}[Collision function] \label{def:gapcol}
    Let $\gamma \in (0,1)$.
    The  function $\gapcol^{k,\gamma}_{N,R}: D_{\gapcol^{k,\gamma}_{N,R}} \subset [R]^N\to\{-1,1\}$ is defined by $\gapcol^{k,\gamma}_{N,R}(s)=-1$ if no integer occurs at least $k$ times in $s$, $\gapcol^{k,\gamma}_{N,R}(s)=1$ if there are more than $\gamma R$ distinct integers that occur at least $k$ times in $s$, and it is undefined otherwise.
\end{definition}

Notice that this problem is not a property testing problem, as the outcome is not determined based on the distance between inputs. Nevertheless, it is a valid promise problem and a special case of testing $k$-collision-freeness, which we use to prove lower bounds on the other problems of interest.

To prove a bound on the $\gapcol$ function, we will actually relate it to the composition of two more elementary functions $g,h$, where by composition we mean $(g\circ h)(x)=g(h(x_1),\dots,h(x_n))$ (where $x=(x_1,\dots,x_n)$ and each $x_i$ is a binary vector of appropriate dimension).
Let us define (i) the threshold function $\thr_N^k:\{-1,1\}^N\to\{-1,1\}$ which is $-1$ if the input bitstring contains at least $k$ many $-1$s, and it is 1 otherwise; and (2) the gap version of $\mathrm{OR}$, that is $\gapor_R^\gamma: D_{\gapor_R^\gamma} \subset \{-1,1\}^R\to\{-1,1\}$ which takes value 1 if the input is $1^R$, $-1$ if the input contains at least $\gamma R$ many $-1$s, and is undefined otherwise.
We show that the double-promise approximate degree of $\gapor_R^\gamma\circ\thr_N^k$ lower bounds the quantum query complexity of the collision problem.

\begin{lemma} \label{lem:qc_dpdeg_kcol}
Let $k\geq 3$, $0<\gamma<1$, $\delta>0$ and $c>2$ be constants such that $N/c \leq R\leq N/2$. 
If the double-promise $\delta$-approximate degree of $\gapor_R^\gamma\circ\thr_N^k$ on domain further restricted to $H_{\le N}^{NR}$ is at least $\Delta$,
    then every quantum algorithm computing $\gapcol^{k,\gamma/c}_{N,N}$ with error $\delta/2$ must require at least $\Delta/2$ queries.
\end{lemma}

Before proving this lemma, we prove some helper propositions.
In order to apply the dual polynomial method for partial  functions,
we start by proving that $\gapcol^{k,\gamma'}_{N,R'}$  
is at least as hard as a very similar problem.
We introduce a ``dummy-augmented'' version $\dgapcol_{N,R}^{k,\gamma}:D_{\dgapcol_{N,R}^{k,\gamma}} \subseteq [R]_0^N \to \{-1,1\}$ of the problem $\gapcol_{N,R}^{k,\gamma}$ for the purpose of proving \Cref{lem:qc_dpdeg_kcol}, where now the input sequence can have integer $0$, but those $0$s are just ignored when they occur.
We show that it is enough to prove a lower bound for this second version.

\begin{proposition} \label{pro:reduction_dummy}
Let $k\geq 3$, $0<\gamma<1$ and $c>2$ be constants such that $N/c \leq R\leq N/2$. Then $\dgapcol^{k,\gamma}_{N,R}$ can be reduced to $\gapcol^{k,\gamma/c}_{N,N}$.
\end{proposition}

\begin{proof}
    An input to $\dgapcol^{k,\gamma}_{N,R}$ is a sequence $s=(s_1,\dots,s_N)$ where each $s_i\in[R]_0$.
    Let us define a family of functions $T_i$ that map from $[R]_0$ to $[R']$ for $R' = R + \lceil N/2 \rceil$: $T_i(s)=s$ if $s>0$ and 
    $T_i(0)=R+\lceil i/2 \rceil$.

    Notice that $(s_1,\dots,s_N)$ is free from $k$-collisions (ignoring collisions of the dummy character 0) if and only if $(T_1(s_1),\dots,T_N(s_N))$ is free from $k$-collisions, i.e., new $k$-collisions cannot be created by this transformation (only 2-collisions but we assume $k\ge 3$).

    On the other hand, if $(s_1,\dots,s_N)$ contains more than $\gamma R$ distinct $k$-collisions, then so does $(T_1(s_1),\dots,T_N(s_N))$.
    Since $\gamma R \geq (\gamma/c) N$, $\gapcol^{k,\gamma/c}_{N,N}$ will reject.
\end{proof}

The following proposition relates $\dgapcol$ to $\gapor \circ \thr$.

\begin{proposition}\label{bindagpcol}
The domain of $\gapor_R^\gamma\circ\thr_N^k$ is
$$D_{\gapor_R^\gamma\circ\thr_N^k}=\{x\in \{-1,1\}^{NR}: (\thr_N^k(x_1),\dots,\thr_N^k(x_R))\in H_{\ge\gamma R}^R\cup\{1^R\}\}.$$
where $x=(x_1,\dots,x_R)$ with each $x_i\in\{-1,1\}^N$.

The domain of $(\dgapcol^{k,\gamma}_{N,R})^{\leq N}$ is 
$$D_{(\dgapcol^{k,\gamma}_{N,R})^{\leq N}}=H_{\le N}^{NR}\cap D_{\gapor_R^\gamma\circ\thr_N^k}.$$ 

Moreover, restricted to the latter domain they are the same function:
$$(\dgapcol^{k,\gamma}_{N,R})^{\leq N} = \gapor_R^\gamma\circ\thr_N^k.$$
\end{proposition}

We are now ready to give the proof of \cref{lem:qc_dpdeg_kcol}.

\begin{proof}[Proof of \cref{lem:qc_dpdeg_kcol}]
By \Cref{pro:reduction_dummy}, instead of $\gapcol^{k,\gamma/c}_{N,N}$ we can consider $\dgapcol^{k,\gamma}_{N,R}$ (with the appropriate parameters) to show a lower bound.
By \Cref{bindagpcol},
we can use \Cref{lem:QC_deg} 
to relate the query complexity of $\dgapcol^{k,\gamma}_{N,R}$ to 
the double-promise degree of $\gapor_R^\gamma\circ\thr_N^k$ with domain further restricted to $H_{\le N}^{NR}$. 
\end{proof}

\subsection{Main lower bound}
Let us fix $f=(\gapor_R^\gamma\circ\thr_N^k)$ with 
domain $D=D_{(\gapor_R^\gamma\circ\thr_N^k)}$ (see \Cref{bindagpcol}). 
For technical reasons,
in the rest of the section we fix $k\geq 3$ and $N=\lceil 20 (2k)^{k/2}\rceil R$.\footnote{These parameters are used in \cite{polynomial_strikes_back} to prove \Cref{pro:final_zeta}, which we will use.}

We first define a construction used in order to compose dual polynomials, which was
introduced in earlier line of work~\cite{ShiZ09,Lee09,Sher13}.
\begin{definition}[Dual block composition] \label{def:dual_block_comp}
    The dual block composition of two functions $\phi:\{-1,1\}^n\to\mathbb{R}$ and $\psi:\{-1,1\}^m\to\mathbb{R}$ is a function $\phi\star \psi:\{-1,1\}^{nm}\to\mathbb{R}$ defined as 
    $$(\phi\star \psi)(x)=2^n \,\phi(\sgn(\psi(x_1)),\dots,\sgn(\psi(x_n)))\prod_{i\in[n]}|\psi(x_i)|$$
    where $x=(x_1,\dots,x_n)$ and $ x_i\in\{-1,1\}^m$, for $i\in[n]$.
\end{definition}

This subsection is dedicated to the proof of the following lemma which, 
together with \Cref{lem:qc_dpdeg_kcol}, implies \Cref{lowerbound-col}.
Observe that we have to zero out the support of the dual polynomial outside of $H_{\le N}^{NR}$, since our target domain is not $D$ but $D\cap H_{\le N}^{NR}$ in
\Cref{lem:qc_dpdeg_kcol}.
\begin{lemma} \label{lem:zeta_dpdeg}
Let $N=\lceil 20 (2k)^{k/2}\rceil R$ and $0<\gamma<1/4^{k-1}$.
Then there exists a function $\zeta:\{-1,1\}^{NR}\to\mathbb{R}$ such that
    \begin{align}
\forall x\in\{-1,1\}^{NR}\setminus H_{\le N}^{NR},\quad \zeta(x)=0 \label{zero};\\
 \sum_{x\in H_{\le N}^{NR} \cap D}\zeta(x)f(x)
 -\sum_{x\in H_{\le N}^{NR}\setminus D}|\zeta(x)|>2/3;  \label{eps}
 \\
   \|\zeta\|_1=1 \quad \text{ and } \quad
\phd(\zeta)\in \Omega\left(\sqrt{N^{1-1/k}}/\ln^2 N\right).  \label{norm-phd}
    \end{align}
\end{lemma}

\begin{proof}
The construction of $\zeta$ starts by \emph{block composing} (\Cref{def:dual_block_comp})
two dual polynomials $\phi,\psi$,
one for $\gapor_R^\gamma$ and one for $\thr_N^k$. 
The dual polynomial $\phi$ for $\gapor_R^\gamma$ is given by \Cref{pro:phi}.
The dual polynomial $\psi$ for $\thr_N^k$ is given by the first part of \Cref{pro:zeta_merged}. 

The block composition $\phi \star \psi$ is a good candidate for the dual polynomial of $f$.
Indeed, \Cref{claim:second_point} shows that it satisfies \Cref{eps}, showing correlation at least $9/10 > 2/3$. One could also check that it satisfies \Cref{norm-phd}. Nonetheless it does not satisfy \Cref{zero}.

We can now use the second part of \Cref{pro:zeta_merged} to argue that there exists another dual polynomial $\zeta$ that satisfies \Cref{zero} and \Cref{norm-phd}.
Moreover, this $\zeta$ is close to $\phi\star\psi$ so that it also satisfies \Cref{eps}, with the weaker but sufficient correlation $9/10-2/9>2/3$.
This concludes the proof.
\end{proof}

As we have seen, the previous proof relies on the following results.
The first one is direct and we omit its proof.
\begin{proposition} \label{pro:phi}
Let $\phi:\{-1,1\}^R\to\mathbb{R}$ be such that
$\phi(-1^R)=-1/2$, $\phi(1^R)=1/2$, and $\phi(z)=0$ for all $z\in\{-1,1\}^R\setminus\{-1^R, 1^R\}$.
Then
$\|\phi\|_1=1$, $\textnormal{phd}(\phi)\geq 1$,
and  $$\sum_{x\in\{-1,1\}^R}
\phi(x)\mathrm{OR}(x)=1.$$ 
\end{proposition}

The second lemma is the rephrasing of several scattered results in~\cite{polynomial_strikes_back} that we unify into one statement for more clarity. For the sake of completeness, we explain how to prove it in \Cref{app:lemma_proof}.
\begin{restatable}{lemma}{lemmaun} \label{pro:zeta_merged}
Let $N=\lceil 20 (2k)^{k/2}\rceil R$ and $\phi:\{-1,1\}^R\to\mathbb{R}$ from \Cref{pro:phi}.
Then there exist $\psi:\{-1,1\}^N\to\mathbb{R}$ and $\zeta:\{-1,1\}^{NR}\to\mathbb{R}$ such that
\begin{enumerate}
\item  
$\|\psi\|_1=1$, $\phd(\psi)\ge c_1\sqrt{k^{-1}N^{1-1/k}}$,
$$\sum_{x\in D_+}|\psi(x)|\le\frac{1}{48N}, \quad \text{ and } \quad 
\sum_{x\in D_-}|\psi(x)|\le\frac{1}{2}-\frac{2}{4^k},$$
where $D_+=\{x\in\{-1,1\}^N: \psi(x)>0, \thr_N^k(x)=-1\}$ and $D_-=\{x\in\{-1,1\}^N: \psi(x)<0, \thr_N^k(x)=1\}$. 
\item
$\|\zeta\|_1=1$, $\phd(\zeta)=\Omega(\sqrt{N^{1-1/k}}/\ln^2N)$, $\|\zeta-\phi\star\psi\|_1\le 2/9$, and $\zeta(x)=0$ for all  $x\in\{-1,1\}^{NR}\setminus H_{\le N}^{NR}$.
\end{enumerate}
\end{restatable}

For the next lemma, we will use the following proposition, which was implicitly used
in the proofs of \cite[Propositions 5.5 and 5.6]{polynomial_strikes_back}, but not stated in this general form. We include its proof in \Cref{app:pro_proof}. By convention, we denote $D_{+1}=D_+$ and $D_{-1}=D_-$.

\begin{restatable}{proposition}{propun}
\label{pro:like_5.5_5.6} 
Let $S\subseteq\{-1,1\}^{NR}$.
Let  $g:\{-1,1\}^R\to\{-1,1\}$, $h:\{-1,1\}^N\to\{-1,1\}$,
 $\phi:\{-1,1\}^R\to\mathbb{R}$.
Let $\psi:\{-1,1\}^N\to\mathbb{R}$ be such that
$\|\psi\|_1=1$ and
$\sum_{x\in\{-1,1\}^N} \psi(x)=0$.
Then the following hold.
\begin{enumerate}
    \item When $\lambda$ denotes the probability mass function $\lambda(u)=|\psi(u)|$:
    $$\sum_{x\in S}|(\phi\star\psi)(x)|= 
\sum_{z\in\{-1,1\}^R} |\phi(z)| \cdot \Pr_{x\sim\lambda^{\otimes R}}[x\in S|(\dots,\sgn(\psi(x_i)),\dots)=z];$$
    \item When $\mu^{z_i}_i$ denotes the probability mass function on $\{-1,1\}$ (parameterized by $z_i \in \{-1,1\}$) such that $\mu^{z_i}_i(-1)=2\sum_{x\in D_{z_i}}|\psi(x)|$, and $\mu = \mu^z=\mu^{z_1}_1\otimes\ldots\otimes\mu^{z_R}_R$ the independent product distribution on $\{-1,1\}^R$: 
    $$\sum_{x\in \{-1,1\}^{NR}}(\phi\star\psi)(x)\cdot(g\circ h)(x)= 
\sum_{z\in\{-1,1\}^R} \phi(z) \cdot 
\E_{y\sim\mu}[g(\dots,y_iz_i,\dots)].$$
\end{enumerate}
\end{restatable}

Finally we are ready to prove
the last missing statement, which is our main technical contribution to this part.
The proof of \cite[Lemma 6.9]{polynomial_strikes_back} does not apply directly to this problem: they use the fact that the dual polynomial $\psi$ of their inner function (OR) has one sided error, which is not the case here.

As now we focus on the composed function $f$ (and the dual composition $\phi\star\psi$), the domain is not restricted to small Hamming weight inputs anymore. 
\begin{restatable}{lemma}{lemmadeux} \label{claim:second_point}
Let $N=\lceil 20 (2k)^{k/2}\rceil R$ and $0<\gamma<1/4^{k-1}$. Functions $\phi$ from \Cref{pro:phi} and $\psi$ from \Cref{pro:zeta_merged} satisfy
$$\sum_{x\in D}(\phi\star\psi)(x)\cdot f(x)-\sum_{x\in\{-1,1\}^{NR}\setminus D}|(\phi\star\psi)(x)|\ge 9/10.$$
\end{restatable}
\begin{proof}
We rewrite the left hand side by manipulating the sets we consider in the sums, and then we will bound separately the terms we get.
    \begin{align*}
    &\sum_{x\in D}(\phi\star\psi)(x)\cdot f(x)-\sum_{x\in\{-1,1\}^{NR}\setminus D}|(\phi\star\psi)(x)| \\
    =&\sum_{x\in \{-1,1\}^{NR}}(\phi\star\psi)(x)\cdot (\textnormal{OR}\circ\thr_N^k)(x) \\
    &-\left(\sum_{x\in \{-1,1\}^{NR}\setminus D}(\phi\star\psi)(x)\cdot (\textnormal{OR}\circ\thr_N^k)(x)+\sum_{x\in\{-1,1\}^{NR}\setminus D}|(\phi\star\psi)(x)|\right) \\
    \ge& \sum_{x\in \{-1,1\}^{NR}}(\phi\star\psi)(x)\cdot (\textnormal{OR}\circ\thr_N^k)(x) -2 \sum_{x\in\{-1,1\}^{NR}\setminus D}|(\phi\star\psi)(x)|
    \end{align*}

We first lower bound the first term.
\begin{claim}
    $$ \sum_{x\in \{-1,1\}^{NR}}(\phi\star\psi)(x)\cdot (\textnormal{OR}\circ\thr_N^k)(x)\geq 1-e^{-\frac{R}{4^{k-1}}}-\frac{R}{48N}.$$
\end{claim}
\begin{claimproof}Using Item 2 of \Cref{pro:like_5.5_5.6}, the left hand side can be written as 
    $$\sum_{z\in\{-1,1\}^R}\phi(z)\cdot \E_{y\sim\mu}[\mathrm{OR}(\dots,y_iz_i,\dots)].$$
    Recall that $\phi(z)=0$ when $z$ is anything but $-1^R$ or $1^R$, so only two terms are left to study.

    If $z=-1^R$, using Item 1 of \Cref{pro:zeta_merged} each $y_i$ is $-1$ with probability $\le 1-1/4^{k-1}$ and 1 with probability $\ge 1/4^{k-1}$. If there is any $y_i=1$ then the value of the OR is still $-1$. The probability of this event is $\ge 1-(1-1/4^{k-1})^R\ge 1-e^{-\frac{R}{4^{k-1}}}$. So the expected value is $\le (-1)(1-e^{-\frac{R}{4^{k-1}}})+e^{-\frac{R}{4^{k-1}}}=-1+2e^{-\frac{R}{4^{k-1}}}$. Since in this case 
    $\phi(-1^R)=-1/2$, the contribution to the sum is at most $1/2-e^{-\frac{R}{4^{k-1}}}$.

    If $z=1^R$, then, using Item 1 of \Cref{pro:zeta_merged} again, each $y_i$ is $-1$ with probability $\le 1/(48N)$. If any $y_i$ is $-1$ then the value of the OR becomes $-1$. The union bound tells us that the probability of this is $\le R/(48N)$, so the expected value is at least $ -R/(48N)+1-R/(48N)=1-R/(24N)$. Multiplied by $\phi(1^R)=1/2$ the contribution is at least $ 1/2-R/(48N)$.
    Thus, the first term can be lower bounded by $1-e^{-\frac{R}{4^{k-1}}}-\frac{R}{48N}$.
\end{claimproof}

Now we bound the second term.
\begin{claim}
$$2\sum_{x\in\{-1,1\}^{NR}\setminus D}|(\phi\star\psi)(x)| <e^{-2R\left(\frac{1}{4^{k-1}}-\gamma\right)^2}.$$
\end{claim}
\begin{claimproof}By Item 1 of \Cref{pro:like_5.5_5.6} with $S=\{-1,1\}^{NR}\setminus D$, the term can be written as follows,
    \begin{equation*}
        2\sum_{x\in\{-1,1\}^{NR}\setminus D}|(\phi\star\psi)(x)|
        =2\sum_{z\in\{-1,1\}^R}|\phi(z)|\cdot\Pr_{x\sim\lambda^{\otimes R}}[x\notin D \mid (\dots,\sgn(\psi(x_i)),\dots)=z],
        \end{equation*} 
which, using that $|\phi(z)|=1/2$ when $z$ is $-1^R$ or $1^R$ and $0$ otherwise, collapses to
        \begin{equation*}
        =\Pr_{x\sim\lambda^{\otimes R}}[x\notin D \mid (\dots,\sgn(\psi(x_i)),z\dots)=-1^R] 
        + \Pr_{x\sim\lambda^{\otimes R}}[x\notin D \mid (\dots,\sgn(\psi(x_i)),\dots)=1^R].
    \end{equation*}
    
    In order to bound these two terms we introduce $0/1$-variables $r_i$ and $q_i$, for $i\in[R]$, related to the false positive and false negative inputs.
Define  $r_i=1$ if $\thr_N^k(x_i)=-1$ and $\sgn(\psi(x_i))=1$, and otherwise $r_i=0$. Similarly, $q_i=1$ if $\thr_N^k(x_i)=1$ and $\sgn(\psi(x_i))=-1$, and otherwise $q_i=0$.

    Let us focus on the first term.
    If we sample $x_i$ from the conditional distribution $(\lambda|\sgn(\psi(x_i))=1)$, then $\Pr[r_i=1]=\Pr[\thr_N^k(x_i)=-1| \sgn(\psi(x_i))=1]=2\sum_{x_i\in D_+}|\psi(x_i)|\le 1/(24N)$ (in the last step we used Item 1 of \Cref{pro:zeta_merged}).
    Thus we can upper bounded the probability that an input does not satisfy the promise of GapOR$_R^\gamma$ (i.e. that it is not in $D$) knowing that all the predictions are 1. It means that it contains at least 1 but less than $\gamma R$ many $-1$s, so this many predictions are false positive, which can be expressed by the $r_i$ variables. In the last step below we use the union bound.
 $$       \Pr\Big[x\notin D \mid \forall i\in[R]\ \sgn(\psi(x_i))=1\Big] = \Pr\left[1\le\sum_{i\in[R]}r_i<\gamma R\right] 
        \le \Pr\left[1\le\sum_{i\in[R]}r_i\right]\le \frac{R}{24N}.
   $$

Similarly, for the second term, if we sample $x_i$ from the conditional distribution $(\lambda|\sgn(\psi(x_i))=-1)$, then $\Pr[q_i=1]=\Pr[\thr_N^k(x_i)=1| \sgn(\psi(x_i))=-1]=2\sum_{x_i\in D_-}|\psi(x_i)|\le 1-\frac{1}{4^{k-1}}$
    (for the last step we used Item 1 of \Cref{pro:zeta_merged} again).

    Then, similarly to the first term, we can upper bound the probability. Now in the last step we use the Chernoff bound, which introduces the constraint $\gamma<\frac{1}{4^{k-1}}$.
    \begin{equation*}
    \Pr\Big[x\notin D \mid \forall i\in[R]\ \sgn(\psi(x_i))=-1\Big] 
    \le \Pr\left[(1-\gamma)R<\sum_{i\in[R]}q_i\right]< e^{-2R\left(\frac{1}{4^{k-1}}-\gamma\right)^2}. 
    \end{equation*}
\end{claimproof}

Putting together the two bounds, we obtain
$$ \sum_{x\in D}(\phi\star\psi)(x)\cdot f(x)-\sum_{x\in\{-1,1\}^{NR}\setminus D}|(\phi\star\psi)(x)|
        \ge 1-\frac{R}{16N}-e^{-\frac{R}{4^{k-1}}}-e^{-2R(\frac{1}{4^{k-1}}-\gamma)^2}.$$
    When $k$ and $1/4^{k-1}-\gamma$ are positive constants and $R\in\Theta(N)$, this is larger than $9/10$ (for large enough $N$).
\end{proof}

Finally, we can conclude the proof of \Cref{lowerbound-col}.
\begin{proof}[Proof of \Cref{lowerbound-col}] \label{proof:th_main_lb}
    By \Cref{lem:zeta_dpdeg}, there is a dual polynomial for $\gapor_R^\gamma\circ \thr_N^k$ of pure high degree $\Omega(\sqrt{N^{1-1/k}}/\ln^2 N)$ that is only supported on $H_{\le N}^{NR}$. By \Cref{th:dpdeg0} this means that the double-promise $\delta$-approximate degree of $\gapor_R^\gamma\circ \thr_N^k$ with domain restricted to $H_{\le N}^{NR}$ is $\Omega(\sqrt{N^{1-1/k}}/\ln^2 N)$. Using \Cref{lem:qc_dpdeg_kcol}
    with $c=\lceil 20 (2k)^{k/2}\rceil$,
    we obtain that the bounded-error quantum query complexity of $\gapcol_{N,N}^{k,\gamma'}$ is $\Omega(\sqrt{N^{1-1/k}}/\ln^2 N)$ if $\gamma'=\gamma/c<1/(4^{k-1}\lceil 20 (2k)^{k/2}\rceil)$. This implies the same lower bound on testing $k$-collision-freeness with $\eps=\gamma'$ as $\gapcol$ is just a more restricted version of the same problem.
\end{proof}

%% file: arxiv_version_3/other_problems_arxiv.tex
\section{Testing $3$-colorability}
\label{sec:quantum_3_coloability_lb}

In this section, we will prove that the problem of testing $3$-colorability in bounded degree graphs remains maximally hard-to-test in the quantum setting. Our lower bound proof will roughly follow the same approach as that of \cite{bogdanov2002lower}, see \cite[Section 5.6]{bhattacharyya2022property} also for a reference. Our result is stated below.

\begin{theorem}[Restatement of \Cref{theo:3color_lb_intro}]\label{theo:3colorlbmain}
Let $G$ be an unknown undirected $N$-vertex graph with maximum degree $d$, and $\eps \in (0,1)$ be a parameter. Given quantum query access to $G$ in the undirected bounded-degree graph model,
in order to distinguish if $G$ is $3$-colorable, or if it is $\eps$-far from being $3$-colorable, $\Omega(N)$ quantum queries are necessary.    
\end{theorem}

In order to prove the above theorem, we will first discuss the approach to proving the classical lower bound. Then we will modify the classical proof suitably to the quantum setting. Let us start with the notion of $k$-wise independent string which will be used both in the classical and quantum lower bound proofs.

\begin{definition}[$k$-wise independent string]
A string $s=(s_1, \ldots, s_N) \in \{0,1\}^N$ is said to be \emph{$k$-wise independent} if for any set of $k$-indices $i_1, i_2 \ldots, i_k$, the probability of any particular assignment  $(b_{i_1},b_{i_2}, \ldots, b_{i_k}) \in \{0,1\}^k$ to the indices $i_1, i_2 \ldots, i_k$ is equal to $1/2^k$. 
\end{definition}

\subsection{Classical lower bound approach for testing $3$-colorability}

As we discussed in the overview, to prove the lower bound of $3$-colorability, the authors in \cite{bogdanov2002lower} studied another problem called E$(3, c)$LIN-2, a problem related to deciding the satisfiability of a system of linear equations. Then the authors designed a reduction to $3$-colorability from E$(3, c)$LIN-2, which finally proves the linear query complexity lower bound for testing $3$-colorability. We will also follow a similar approach here. Let us first formally define the problem of E$(3, c)$LIN-2, where below $\mathbb{F}_2$ denotes the field with two elements.

\begin{definition}[E$(3, c)$LIN-2]
Let $\cE$ be a system of linear equations with $N$ variables from $\mathbb{F}_2$, where there are $3$ variables in each equation, and each variable occurs in at most $c$ equations. This system $\cE$ is represented as a matrix and we have query access to its entries. Given a parameter $\alpha \in (0,1)$, the goal is to distinguish if $\cE$ is satisfiable, or at least an $\alpha$-fraction of the equations need to be modified to make $\cE$ satisfiable.
\end{definition}

The authors in \cite{bogdanov2002lower} proved the following lemma, which states that there exists a system of linear equations (equivalently a matrix), such that any constant fraction of the rows of this matrix are linearly independent. The authors proved this using hypergraph constructions.

\begin{lemma}[{\cite[Theorem 8]{bogdanov2002lower}}]\label{lem:existhardmatrix}
For every $c > 0$, there exists a $\delta > 0$ such that
for every $N$, there exists a matrix $A \in \{0, 1\}^{cN \times N}$ with $cN$ rows and $N$
columns such that the following conditions hold:
\begin{enumerate}
    \item Each row of $A$ has exactly three non-zero entries;
    \item Each column of $A$ has exactly $3c$ non-zero entries;
    \item Every collection of $\delta N$ rows of $A$ is linearly independent.
\end{enumerate} 
\end{lemma}

Using the existence of the matrix $A$ corresponding to \Cref{lem:existhardmatrix}, the authors in \cite{bogdanov2002lower} used Yao's minimax lower bound technique to prove a linear lower bound for testing E$(3, c)$LIN-2. More formally, they designed a pair of hard-to-distinguish distributions $D_{\mathrm{yes}}$ and $D_{\mathrm{no}}$, such that unless $\Omega(N)$ queries are performed, no algorithm can distinguish between them.

\begin{restatable}{lemma}{linmatrixlb}
\label{lem:hardnesse3lin2}
There exists a matrix $A \in \{0, 1\}^{cN \times N}$ (similar to the matrix mentioned in \Cref{lem:existhardmatrix}) such that given a parameter $\eps \in (0,1)$ and query access to $A$ and a vector $y\in\{0,1\}^{cN}$, in order to distinguish if there exists another vector $x \in \{0,1\}^N$ such that $Ax=y$, or for any vector $x \in \{0,1\}^N$, only a constant $\eps$ fraction of the constraints encoded by $A$ and $y$ are satisfied, $\Omega(N)$ queries are necessary. 
\end{restatable}

For completeness, we will briefly describe the proof of the above lemma in \Cref{sec:3colorability_lb_app}. The authors in \cite{bogdanov2002lower} proved the following lower bound for testing E$(3, c)$LIN-2.

\begin{lemma}[{\cite[Lemma 19]{bogdanov2002lower}}]\label{lem:3colorindistinguishability}
For every $\alpha > 0$, there are constants $c$ and $\delta > 0$ such that every algorithm that distinguishes satisfiable instances of E$(3, c)$LIN-2 with $N$ variables from instances that are $(1/2 - \alpha)$-far from satisfiable must have classical query complexity at least $\delta N$.    
\end{lemma}

The key insight that is used to prove the above lemma is that applying \Cref{lem:existhardmatrix}, any $\delta N$ rows of $A$ are linearly independent, thus any subset of $\delta N$ entries of $Ay=z$ will look uniformly random. Hence $z$ is ``\emph{$k$-wise independent}'' with $k=\delta N$.
We will not formally prove the above lemma here, please refer to \cite{bogdanov2002lower} for a formal proof.

Finally, we have the reduction that maps satisfying instances of testing E$(3, c)$LIN-2 to satisfying instances of testing $ 3$-colorability and vice-versa. 

\begin{lemma}[{\cite[Section 4]{bogdanov2002lower}}]\label{lem:e3lin3colrreduction}
There exists a reduction $\varphi$ that maps instances of testing E$(3, c)$LIN-2 to instances of testing 3-colorability such that the following hold:

\begin{enumerate}
    \item If an input $x$ to E$(3, c)$LIN-2 is satisfiable, then $\varphi(x)$ is a $3$-colorable graph;

    \item If an input $x$ to E$(3, c)$LIN-2 is far from being satisfiable, then $\varphi(x)$ is a graph that is far from being $3$-colorable.
\end{enumerate} 
\end{lemma}

\subsection{Quantum lower bound for testing E$(3, c)$LIN-2 and $3$-colorability}

We will first prove the quantum lower bound for testing E$(3, c)$LIN-2. Our result is stated as follows.

\begin{lemma}\label{lem:e3linqunatumlb}
For every $\alpha > 0$ there are constants $c$ and $\delta > 0$ such that every algorithm that distinguishes satisfiable instances of E$(3, c)$LIN-2 with $N$ variables from instances that are $(1/2 - \alpha)$-far from
satisfiable must have quantum query complexity at least $\delta N/2$.    
\end{lemma}

In order to prove the above lemma, we will be using the following observation from \cite{apers2022quantum}, which states that distinguishing between a uniformly random string and a $\ell$-wise independent string, for an appropriate integer $\ell$, is hard for quantum algorithms.

\begin{proposition}[Fact 1 from \cite{apers2022quantum}]\label{obs:kindependnethardness}
The output distribution of a quantum algorithm making $q$ queries to a uniformly random
string is identical to the same algorithm making $q$ queries to a $2q$-wise independent string.   
\end{proposition}

Now let us prove \Cref{lem:e3linqunatumlb}.

\begin{proof}[Proof of \Cref{lem:e3linqunatumlb}]
Following \Cref{lem:existhardmatrix,lem:hardnesse3lin2}, we know that there exists a matrix $A$ whose $\delta N$ rows are linearly independent, for which testing E$(3, c)$LIN-2 requires $\Omega(N)$ classical queries. Moreover, from \Cref{obs:kindependnethardness}, we know that any quantum algorithm that performs less than $k/2$ queries can not distinguish a uniformly random vector from a $k$-wise independent vector. Now let us  set $k=\delta N$. Combining all the above, this implies that at least $\delta N/2$  quantum queries are necessary for testing E$(3, c)$LIN-2.
\end{proof}

Now we are finally ready to prove \Cref{theo:3colorlbmain}.

\begin{proof}[Proof of \Cref{theo:3colorlbmain}]
From \Cref{lem:e3linqunatumlb}, we know that the quantum query complexity of testing E$(3, c)$LIN-2 is $\Omega(N)$. In order to prove similar lower bound for testing $3$-colorability, we will again use a reduction approach. Given a pair of hard instances corresponding to testing E$(3, c)$LIN-2, we will apply the reduction $\psi$ mentioned in \Cref{lem:e3lin3colrreduction}. Similar to the classical setting, $\varphi$ will map the yes instances of E$(3, c)$LIN-2 to instances of $3$-colorable graphs and vice-versa. So, the quantum query lower bound of $\Omega(N)$ carries forward from E$(3, c)$LIN-2 to $3$-colorability. Thus, we conclude that $\Omega(N)$ quantum queries are necessary to test $3$-colorability in the bounded degree model.
\end{proof}

\subsection{Other maximal hard-to-test problems}

As we mentioned in the introduction, there are several other problems in the bounded degree graph model, which are maximally hard-to-test. Moreover, their lower bounds stem from similar ideas as the E$(3, c)$LIN-2 and $3$-colorability lower bounds, as mentioned in \cite{yoshida2010query, goldreich2020testing}. Following the same path as in the previous subsection, we also obtain $\Omega(N)$ quantum query lower bounds for all these problems. For brevity, we only present the theorem statements below and omit their proofs.

\begin{theorem}[Hamiltonian Path/Cycle]
Given quantum query access to an unknown undirected (directed) $d$-bounded degree $N$-vertex graph $G$ for some integer $d$, and a parameter $\eps \in (0,1)$, in order to distinguish if $G$ has an undirected (directed) Hamiltonian path/cycle or $\eps$-far from having an undirected (directed) Hamiltonian path/cycle,  $\Omega(N)$ quantum queries are necessary. 
\end{theorem}

\begin{theorem}[Approximating Independent Set/Vertex Cover size]
Given query access to an unknown undirected $d$-bounded degree $N$-vertex graph $G$ for some integer $d$, and a parameter $\eps \in (0,1)$, approximating the independent set size/vertex cover of $G$, $\Omega(N)$ quantum queries are necessary.
\end{theorem}

%% file: arxiv_version_3/conclusion_arxiv.tex
\section{Conclusion}
We provided new algorithms and lower bounds in the so far unexplored field of quantum property testing of directed bounded degree graphs. On the way, we revisited the well-known problem of collision finding in a new, property testing setting.

In particular, we used the dual polynomial method to obtain the first step for adapting the proportional moments technique of \cite{raskhodnikova2009strong} (for proving randomized lower bounds) to the quantum setting.  
Indeed, the classical lower bounds of \cite{hellweg2013property} and \cite{peng2023optimal} for testing subgraph-freeness use
the results of \cite{raskhodnikova2009strong} for a collision-related problem. Recently, the authors in \cite{ambainis2016efficient} stated it as an open question if one could use a variant of the proportional moments technique of \cite{raskhodnikova2009strong} for proving better quantum lower bounds. This remains an interesting open question, 
but we hope that this work will serve as a step towards obtaining this new quantum lower bound technique.

%% file: arxiv_version_3/proofs_appendix_arxiv.tex
\section{Deferred material from \Cref{sec:quantum_lb}} \label{appendix:proofs}

\subsection{Dual polynomial for THR} \label{app:psi}
\begin{definition}
    Let $M\in\mathbb{N}$ and $\alpha,\beta>0$. We say that a function $\omega:[M]_0\to\mathbb{R}$ satisfies the $(\alpha,\beta)$-decay condition if $\sum_{t\in[M]_0}\omega(t)=0$, $\sum_{t\in[M]_0}|\omega(t)|=1$ and $|\omega(t)|\le \alpha e^{-\beta t}/t^2$.
\end{definition}

In \cite[Section 5.1]{polynomial_strikes_back} the authors define a dual polynomial $\psi$ of $\thr_N^k$ in the following way.
Let $k,N\in\mathbb{N}$, and $T$ an integer such that $k\le T\le N$.
Let $c=2k \lceil N^{1/k}\rceil$ and $m=\lfloor \sqrt{T/c}\rfloor$.
Define set $S=\{1,2,\dots,k\}\cup\{ci^2:0\le i\le m\}$.
Define a univariate polynomial $$\omega(t)=\frac{(-1)^{t+T-m+1}}{T!}\binom{T}{t}\prod_{r\in[T]_0\setminus S}(t-r).$$
Then let $\psi:\{-1,1\}^N\to\mathbb{R}$ be $\psi(x)=\omega(|x|)/\binom{N}{|x|}$ for $x\in H_{\le T}^N$ and $\psi(x)=0$ otherwise.

They show that $\psi$ and $\omega$ have the following properties.

\begin{proposition}{\cite[Proposition 5.4]{polynomial_strikes_back}} \label{pro:psi_properties_full}
    Let $\omega$ and $\psi$ be the polynomials defined above. Then the following are true.
    \begin{enumerate}
        \item $\sum_{x\in D_+}|\psi(x)|\le\frac{1}{48N}$;
        \item $\sum_{x\in D_-}|\psi(x)|\le\frac{1}{2}-\frac{2}{4^k}$;
        \item $\|\psi\|_1=1$;
        \item $\textnormal{phd}(\psi)\ge c_1\sqrt{k^{-1}TN^{-1/k}}$;
        \item $\omega$ satisfies the $(\alpha,\beta)$-decay condition with $\alpha=(2k)^k$ and $\beta=c_2/\sqrt{kTN^{1/k}}$.
    \end{enumerate}
\end{proposition}

Here $D_+$ and $D_-$ denote the set of false positives and that of false negatives respectively if $\psi$ is considered as a hypothesis for $\thr_N^k$, i.e. $D_+=\{x\in\{-1,1\}^N: \psi(x)>0, \thr_N^k(x)=-1\}$ and $D_-=\{x\in\{-1,1\}^N: \psi(x)<0, \thr_N^k(x)=1\}$. 

\subsection{Proof of \Cref{pro:zeta_merged}} \label{app:lemma_proof}
We start with two propositions from \cite{polynomial_strikes_back} that we are going to use.
The first one is about the properties of the dual block composition. 
\begin{proposition}{\cite[Proposition 2.20]{polynomial_strikes_back}} \label{pro:dual_compose_props}
    Let $\phi:\{-1,1\}^n\to\mathbb{R}$, $\psi:\{-1,1\}^m\to\mathbb{R}$.
    The dual block composition has the following properties.
    \begin{enumerate}
        \item If $\|\phi\|_1=1$, $\|\psi\|_1=1$ and $\braket{\phi,1^m}=0$ then $\|\phi\star \psi\|_1=1$.
        \item If $\textnormal{phd}(\phi)\ge \Delta$ and $\textnormal{phd}(\psi)\ge \Delta'$ then $\textnormal{phd}(\phi\star \psi)\ge \Delta\cdot \Delta'$.
    \end{enumerate}
\end{proposition}

The second one proves the existence of the final dual polynomial $\zeta$, that is close to the ``almost good'' block composition $\phi\star\psi$, given that some conditions are satisfied.
\begin{proposition}{\cite[Proposition 2.22]{polynomial_strikes_back}} \label{pro:final_zeta}
    Let $R\in\mathbb{N}$ sufficiently large and $M\le R$. Let $\phi:\{-1,1\}^R\to\mathbb{R}$ with $\|\phi\|_1=1$, and let $\omega:[M]_0\to\mathbb{R}$ satisfy the $(\alpha,\beta)$-decay condition with some $1\le\alpha \le R^2$ and $4\ln^2R/(\sqrt{\alpha}R)\le\beta\le1$. Let $N=\lceil20\sqrt{\alpha}\rceil R$ and $\psi:\{-1,1\}^N\to\mathbb{R}$ be defined as $\psi(x)=\omega(|x|)/\binom{N}{|x|}$. Let $\Delta<N$ be such that $\textnormal{phd}(\phi\star\psi)\ge \Delta$.
    Then there exist a $\Delta'\ge\beta\sqrt{\alpha}R/(4\ln^2R)$ and a function $\zeta:\{-1,1\}^{NR}\to\mathbb{R}$ such that
    \begin{enumerate}
        \item $\textnormal{phd}(\zeta)\ge\min\{\Delta,\Delta'\}$;
        \item $\|\zeta-\phi\star\psi\|_1\le2/9$;
        \item $\|\zeta\|_1=1$;
        \item $\forall x\in\{-1,1\}^{NR}$ with $|x|>N$ $\zeta(x)=0$.
    \end{enumerate}
\end{proposition}

We now restate \Cref{pro:zeta_merged} before proving it.
\lemmaun*

\begin{proof} 
With the construction of $\psi$ described in \Cref{app:psi}, \Cref{pro:psi_properties_full} (with $T=N$) ensures that Item 1 is satisfied.
This way, we know that $\|\psi\|_1=1$, and that $\phd(\psi)\ge c_1\sqrt{k^{-1}N^{1-1/k}}$.
From \Cref{pro:phi} we know that $\|\phi\|_1=1$ and $\phd(\phi)\ge 1$.
Using item 1 of \Cref{pro:dual_compose_props} we obtain $\|\phi\star\psi\|_1=1$, and using Item 2 we get $\phd(\phi\star\psi)\ge c_1\sqrt{k^{-1}N^{1-1/k}}$.

From \Cref{pro:psi_properties_full} we know that the function $\omega$ that is used to define $\psi$ satisfies the $(\alpha,\beta)$-decay condition for some constant $\alpha=(2k)^k$ and $\beta=c_2/\sqrt{kN^{1+1/k}}$.

This way, we can use \Cref{pro:final_zeta} to obtain the function $\zeta$ we wanted for Item 2.
Indeed, our functions $\psi$ and $\phi$ satisfy all the conditions of the lemma with pure high degree lower bounded by $\Delta=c_1\sqrt{k^{-1}N^{1-1/k}}$; and with our parameters $\alpha$, $\beta$ we obtain $\Delta'=c_2 (2k)^{k/2} R/(4\ln^2(R) \sqrt{kN^{1+1/k}}) \in\Omega(\sqrt{N^{1-1/k}}/\ln^2 N)$.
\end{proof}

\subsection{Proof of \Cref{pro:like_5.5_5.6}} \label{app:pro_proof}
Let us restate the proposition that we are going to prove.
\propun*

\begin{proof}
Remember that $\lambda$ denotes the probability mass function $\lambda(u)=|\psi(u)|$ for $u\in\{-1,1\}^N$. We will need the following claim.
\begin{claim} \label{claim:psi_plus_minus}
$$\Pr_{u\sim \lambda}[\psi(u)>0]=\Pr_{u\sim \lambda}[\psi(u)<0]=\frac{1}{2}.$$
\end{claim}
\begin{claimproof}We know that $\sum_u \psi(u)=0$. Thus $\sum_{u:\psi(u)>0}|\psi(u)|-\sum_{u:\psi(u)>0}|\psi(u)|$. We then conclude using that $\|\psi\|_1=1$.
\end{claimproof}

\paragraph*{First part of \Cref{pro:like_5.5_5.6}}
Below, we first apply the definition of the dual block composition (and the fact that $2^R$ and $\prod_{i\in[R]}|\psi(x_i)|$ are positive). Then we use the definition of $\lambda$ which ensures that $\prod_{i\in[R]}|\psi(x_i)|$ is the probability of getting $x=(\dots,x_i,\dots)$ when sampling independently $R$ times from distribution $\lambda$.
\begin{eqnarray*}
    \sum_{x\in S}|(\phi\star\psi)(x)| 
   &= & 2^R\sum_{x\in \{-1,1\}^{NR}}\left(\prod_{i\in[R]}|\psi(x_i)|\right) \cdot |\phi(\dots,\textnormal{sgn}(\psi(x_i)),\dots)|\cdot \mathbb{I}[x\in S] \\
    &= & 2^R\cdot\mathbb{E}_{x\sim\lambda^{\otimes R}}\left[|\phi(\dots,\textnormal{sgn}(\psi(x_i)),\dots)|\cdot \mathbb{I}[x\in S]\right]
\end{eqnarray*}

We introduce new variables $z_i$ that will be compared to $\textnormal{sgn}(\psi(x_i))$.
Using \Cref{claim:psi_plus_minus}, the probability of picking a $z\in\{-1,1\}^R$ from the uniform distribution such that $z$ corresponds to the vector of the signs is $\frac{1}{2^R}$.
Thus previous term can be rewritten as
\begin{eqnarray*}
    &&2^R\sum_{z\in\{-1,1\}^R} |\phi(z)| \cdot \Pr_{x\sim\lambda^{\otimes R}}[x\in S\land(\dots,\textnormal{sgn}(\psi(x_i)),\dots)=z] \\
    &=&\sum_{z\in\{-1,1\}^R} |\phi(z)| \cdot \Pr_{x\sim\lambda^{\otimes R}}[x \in S \mid (\dots,\textnormal{sgn}(\psi(x_i)),\dots)=z]
\end{eqnarray*}
which completes the proof.
\paragraph*{Second part of \Cref{pro:like_5.5_5.6}}
Remember that $\lambda$ denotes the probability mass function $\lambda(u)=|\psi(u)|$ for $u\in\{-1,1\}^N$.
    Just like in the proof of the first item, 
    $$\sum_{x\in \{-1,1\}^{NR}}(\phi\star\psi)(x)\cdot(g\circ h)(x)
    =\sum_{z\in\{-1,1\}^R} \phi(z) \cdot 
    \E_{x\sim\lambda^{\otimes R}}\left[(g\circ h)(x) \mid (\dots,\textnormal{sgn}(\psi(x_i)),\dots)=z\right].$$

Using \Cref{claim:psi_plus_minus}, we can first notice that for any $b\in\{-1,1\}$, the probability that an $x_i$ sampled from $\lambda$ is a false $b$ (i.e. false positive if $b=1$ and false negative if $b=-1$) is as follows,
where by convention $D_{+1}=D_+$ and $D_{-1}=D_-$:
\begin{eqnarray*}
    \Pr_{x_i\sim\lambda}\Big[h(x_i)\neq \textnormal{sgn}(\psi(x_i)) \mid \textnormal{sgn}(\psi(x_i))=b\Big]
&=&\sum_{x_i\in D_b} \Pr_{x_i\sim\lambda}\Big[\textnormal{sampling }x_i \mid \textnormal{sgn}(\psi(x_i))=b\Big]\\
&=&2\sum_{x_i\in D_b}|\psi(x_i)|.
\end{eqnarray*}

Therefore, if $z_i=\textnormal{sgn}(\psi(x_i))$ and $x_i$ is a false $z_i$, it means that $z_i$ should be flipped to get $h(x_i)$. Let $y_i\in\{-1,1\}$ denote whether we flip $z_i$. As $x_i$ is a false $z_i$ with probability $2\sum_{x_i\in D_{z_i}}|\psi(x_i)|$, this is the probability with which we should flip $z_i$, i.e. the probability that $y_i=-1$.

Thus for any $z\in\{-1,1\}^R$, the vector $(\dots,h(x_i),\dots)$ with $x\sim\lambda^{\otimes R}$ conditioned on $(\dots,\textnormal{sgn}(\psi(x_i)),\dots)=z$ is identically distributed with $(\dots,z_iy_i,\dots)$ where $y_i$ are random bitflips according to $\mu^{z_i}_i$: $y_i=-1$ with probability $2\sum_{x_i\in D_{z_i}}|\psi(x_i)|$ and $y_i=1$ otherwise.

Now we can finish the proof:
\begin{eqnarray*}
    &&\sum_{z\in\{-1,1\}^R} \phi(z) \cdot 
    \E_{x\sim\lambda^{\otimes R}}\left[(g\circ h)(x) \mid (\dots,\textnormal{sgn}(\psi(x_i)),\dots)=z\right]\\
    &=&\sum_{z\in\{-1,1\}^R} \phi(z) \cdot 
    \E_{y\sim\mu}[g(\dots,z_iy_i,\dots)].
\end{eqnarray*}
\end{proof}

%% file: arxiv_version_3/3colorability_appendix_arxiv.tex
\section{Proofs from \Cref{sec:quantum_3_coloability_lb}: hardness of testing E$(3, c)$LIN-2}\label{sec:3colorability_lb_app}

In this section, we will present a sketch of the proof of the lower bound of testing E$(3, c)$LIN-2. Our main result is stated as follows.

\linmatrixlb*

\begin{proof}[Proof sketch]
	As we mentioned, this proof follows Yao's minimax lower bound technique. A pair of hard distributions $D_{\mathrm{yes}}$ and $D_{\mathrm{no}}$ are constructed, such that, unless $\Omega(N)$ queries are performed, no algorithm can distinguish between them. 

	Let us consider the matrix $A \in \{0, 1\}^{cN \times N}$ as mentioned in \Cref{lem:existhardmatrix}. Based on the matrix $A$, the hard-to-distinguish distributions $D_{\mathrm{yes}}$ and $D_{\mathrm{no}}$ are as follows:
	\begin{enumerate}
		\item \textbf{$D_{\mathrm{yes}}$:} Choose a vector $z \in \{0,1\}^N$ uniformly at random from $\{0,1\}^N$, and set the vector $y\in\{0,1\}^{cN}$ as $y=Az$. Then the system of linear equations is $Ax=y$.
		
		\item \textbf{$D_{\mathrm{no}}$:} Choose the vector $y \in \{0,1\}^{cN}$ uniformly at random from $\{0,1\}^{cN}$, and set the system of linear equations $Ax=y$.
	\end{enumerate}
	
	Now we have the following claim describing the properties of $D_{\mathrm{yes}}$ and $D_{\mathrm{no}}$.
	
	\begin{claim}~
		\begin{enumerate}
			\item[(i)] The system of linear equations corresponding to $D_{\mathrm{yes}}$ is satisfiable.
			
			\item[(ii)] With probability at least $2/3$, the system of linear equations corresponding to $D_{\mathrm{no}}$ is $(1/2 - \alpha)$-far from being satisfiable for every $\alpha > 0$.
		\end{enumerate}
	\end{claim}
	
	Note that the system of linear equations in $D_{\mathrm{yes}}$ is satisfiable by setting $x=z$. On the other hand, for the system of linear equations corresponding to $D_{\mathrm{no}}$, vector $y$ is uniformly random. Thus, with high probability, vector $Az-y$ has large Hamming weight for any $z\in\{0,1\}^N$, and therefore the system of linear equations $Ax=y$ is far from being satisfiable. The formal proof is in \cite[Lemma 18]{bogdanov2002lower}.
\end{proof}

%% file: arxiv_version_3/main_arxiv.bbl
\begin{thebibliography}{}

\end{thebibliography}


\newcommand{\etalchar}[1]{$^{#1}$}
\begin{thebibliography}{BDCG{\etalchar{+}}20}

\bibitem[AAI{\etalchar{+}}16]{AaronsonAIKS16}
Scott Aaronson, Andris Ambainis, Janis Iraids, Martins Kokainis, and Juris Smotrovs.
\newblock Polynomials, quantum query complexity, and {G}rothendieck's inequality.
\newblock In {\em Proceedings of the 31st Conference on Computational Complexity (CCC)}, volume~50, pages 25:1--25:19. Schloss Dagstuhl -- Leibniz-Zentrum f{\"u}r Informatik, 2016.

\bibitem[Aar02]{Aaronson02}
Scott Aaronson.
\newblock Quantum lower bound for the collision problem.
\newblock In {\em Proceedings of 34th Annual {ACM} Symposium on Theory of Computing (STOC)}, pages 635--642. Association for Computing Machinery, 2002.

\bibitem[ABRW16]{ambainis2016efficient}
Andris Ambainis, Aleksandrs Belovs, Oded Regev, and Ronald~de Wolf.
\newblock Efficient quantum algorithms for (gapped) group testing and junta testing.
\newblock In {\em Proceedings of the 27th annual Symposium on Discrete Algorithms (SODA)}, pages 903--922, 2016.

\bibitem[ACL11]{ambainis2011quantum}
Andris Ambainis, Andrew~M Childs, and Yi-Kai Liu.
\newblock Quantum property testing for bounded-degree graphs.
\newblock In {\em Proceedings of the 14th International Workshop and 15th International Conference on Approximation, Randomization, and Combinatorial Optimization: Algorithms and Techniques (APPROX-RANDOM)}, pages 365--376. Springer Berlin Heidelberg, 2011.

\bibitem[ADW22]{apers2022quantum}
Simon Apers and Ronald De~Wolf.
\newblock Quantum speedup for graph sparsification, cut approximation, and {L}aplacian solving.
\newblock {\em SIAM Journal on Computing}, 51(6):1703--1742, 2022.

\bibitem[Amb04]{Ambainis_search}
A.~Ambainis.
\newblock Quantum search algorithms.
\newblock {\em SIGACT News}, 35(2):22–35, 2004.

\bibitem[Amb05]{ambainis_poly}
Andris Ambainis.
\newblock Polynomial degree and lower bounds in quantum complexity: Collision and element distinctness with small range.
\newblock {\em Theory of Computing}, 1(3):37--46, 2005.

\bibitem[AMSS25]{AMSS25_random}
Simon Apers, Fr{\'{e}}d{\'{e}}ric Magniez, Sayantan Sen, and D{\'{a}}niel Szab{\'{o}}.
\newblock Quantum property testing in sparse directed graphs.
\newblock In {\em Approximation, Randomization, and Combinatorial Optimization. Algorithms and Techniques (APPROX/RANDOM)}, volume 353 of {\em LIPIcs}, pages 32:1--32:24. Schloss Dagstuhl - Leibniz-Zentrum f{\"{u}}r Informatik, 2025.

\bibitem[AS04]{AaronsonShi}
Scott Aaronson and Yaoyun Shi.
\newblock Quantum lower bounds for the collision and the element distinctness problems.
\newblock {\em Journal of the ACM}, 51(4):595–605, 2004.

\bibitem[AS19]{DBLP:journals/qic/ApersS19}
Simon Apers and Alain Sarlette.
\newblock Quantum fast-forwarding: Markov chains and graph property testing.
\newblock {\em Quantum Information \& Computation}, 19(3–4):181--213, 2019.

\bibitem[BBC{\etalchar{+}}01]{beals_poly}
Robert Beals, Harry Buhrman, Richard Cleve, Michele Mosca, and Ronald de~Wolf.
\newblock Quantum lower bounds by polynomials.
\newblock {\em Journal of the ACM}, 48(4):778–797, 2001.

\bibitem[BBHT99]{tight_lb_grover}
Michel Boyer, Gilles Brassard, Peter Høyer, and Alain Tapp.
\newblock {\em Tight Bounds on Quantum Searching}, chapter~10, pages 187--199.
\newblock John Wiley \& Sons, Ltd, 1999.

\bibitem[BDCG{\etalchar{+}}20]{ben2020symmetries}
Shalev Ben-David, Andrew~M Childs, Andr{\'a}s Gily{\'e}n, William Kretschmer, Supartha Podder, and Daochen Wang.
\newblock Symmetries, graph properties, and quantum speedups.
\newblock In {\em 61st Annual Symposium on Foundations of Computer Science (FOCS)}, pages 649--660. IEEE, 2020.

\bibitem[BFNR08]{buhrman2008quantum}
Harry Buhrman, Lance Fortnow, Ilan Newman, and Hein R{\"o}hrig.
\newblock Quantum property testing.
\newblock {\em SIAM Journal on Computing}, 37(5):1387--1400, 2008.

\bibitem[BHT98]{BHT}
Gilles Brassard, Peter H{\o}yer, and Alain Tapp.
\newblock Quantum cryptanalysis of hash and claw-free functions.
\newblock In {\em Proceedings of the 3rd Latin American Symposium on Theoretical Informatics (LATIN)}, pages 163--169. Springer Berlin Heidelberg, 1998.

\bibitem[BKT20]{polynomial_strikes_back}
Mark Bun, Robin Kothari, and Justin Thaler.
\newblock The polynomial method strikes back: Tight quantum query bounds via dual polynomials.
\newblock {\em Theory of Computing}, 16(10):1--71, 2020.

\bibitem[BOT02]{bogdanov2002lower}
Andrej Bogdanov, Kenji Obata, and Luca Trevisan.
\newblock A lower bound for testing 3-colorability in bounded-degree graphs.
\newblock In {\em 43rd Annual Symposium on Foundations of Computer Science (FOCS)}, pages 93--102. IEEE Computer Society, 2002.

\bibitem[BR02]{bender2002testing}
Michael~A Bender and Dana Ron.
\newblock Testing properties of directed graphs: acyclicity and connectivity.
\newblock {\em Random Structures \& Algorithms}, 20(2):184--205, 2002.

\bibitem[BT20]{Bun_Thaler_20}
Mark Bun and Justin Thaler.
\newblock A nearly optimal lower bound on the approximate degree of {AC}$^0$.
\newblock {\em SIAM Journal on Computing}, 49(4):FOCS17--59--FOCS17--96, 2020.

\bibitem[BY22]{bhattacharyya2022property}
Arnab Bhattacharyya and Yuichi Yoshida.
\newblock {\em Property Testing - Problems and Techniques}.
\newblock Springer Singapore, 1 edition, 2022.

\bibitem[CPS16]{czumaj2016relating}
Artur Czumaj, Pan Peng, and Christian Sohler.
\newblock Relating two property testing models for bounded degree directed graphs.
\newblock In {\em Proceedings of the 48th annual Symposium on Theory of Computing (STOC)}, pages 1033--1045. Association for Computing Machinery, 2016.

\bibitem[CS10]{Czumaj2010}
Artur Czumaj and Christian Sohler.
\newblock Sublinear-time algorithms.
\newblock In Oded Goldreich, editor, {\em Property Testing: Current Research and Surveys}, pages 41--64. Springer Berlin Heidelberg, 2010.

\bibitem[CY19]{chen2019testability}
Hubie Chen and Yuichi Yoshida.
\newblock Testability of homomorphism inadmissibility: Property testing meets database theory.
\newblock In {\em Proceedings of the 38th ACM SIGMOD-SIGACT-SIGAI Symposium on Principles of Database Systems (PODS)}, pages 365--382. Association for Computing Machinery, 2019.

\bibitem[Fis01]{Fischer01}
Eldar Fischer.
\newblock The art of uninformed decisions.
\newblock {\em Bulletin of the {EATCS}}, 75:97, 2001.

\bibitem[FNY{\etalchar{+}}20]{forster2020computing}
Sebastian Forster, Danupon Nanongkai, Liu Yang, Thatchaphol Saranurak, and Sorrachai Yingchareonthawornchai.
\newblock Computing and testing small connectivity in near-linear time and queries via fast local cut algorithms.
\newblock In {\em Proceedings of the 14th Annual ACM-SIAM Symposium on Discrete Algorithms (SODA)}, pages 2046--2065. Society for Industrial and Applied Mathematics, 2020.

\bibitem[GGR98]{GGR98}
Oded Goldreich, Shari Goldwasser, and Dana Ron.
\newblock Property testing and its connection to learning and approximation.
\newblock {\em Journal of the ACM}, 45(4):653–750, 1998.

\bibitem[Gol17]{goldreich2017introduction}
Oded Goldreich.
\newblock {\em Introduction to Property Testing}.
\newblock Cambridge University Press, 2017.

\bibitem[Gol25]{goldreich2020testing}
Oded Goldreich.
\newblock {\em On Testing Hamiltonicity in the Bounded Degree Graph Model}, pages 282--292.
\newblock Springer Nature Switzerland, 2025.

\bibitem[GR02]{goldreich2002property}
Oded Goldreich and Dana Ron.
\newblock Property testing in bounded degree graphs.
\newblock {\em Algorithmica}, 32(2):302--343, 2002.

\bibitem[Gro96]{Grover}
Lov~K. Grover.
\newblock A fast quantum mechanical algorithm for database search.
\newblock In {\em Proceedings of the 28th Annual ACM Symposium on Theory of Computing (STOC)}, page 212–219. Association for Computing Machinery, 1996.

\bibitem[HLM17]{harrow2017sequential}
Aram~W Harrow, Cedric Yen-Yu Lin, and Ashley Montanaro.
\newblock Sequential measurements, disturbance and property testing.
\newblock In {\em Proceedings of the 28th Annual Symposium on Discrete Algorithms (SODA)}, pages 1598--1611. Society for Industrial and Applied Mathematics, 2017.

\bibitem[HS12]{hellweg2013property}
Frank Hellweg and Christian Sohler.
\newblock Property testing in sparse directed graphs: Strong connectivity and subgraph-freeness.
\newblock In {\em 20th Annual European Symposium on Algorithms (ESA)}, pages 599--610. Springer Berlin Heidelberg, 2012.

\bibitem[Lee09]{Lee09}
Troy Lee.
\newblock A note on the sign degree of formulas.
\newblock {\em arXiv:0909.4607}, 2009.

\bibitem[LZ19]{zhandry_collision}
Qipeng Liu and Mark Zhandry.
\newblock On finding quantum multi-collisions.
\newblock In {\em Advances in Cryptology (EUROCRYPT)}, pages 189--218. Springer International Publishing, 2019.

\bibitem[MdW16]{montanaro2016survey}
Ashley Montanaro and Ronald de~Wolf.
\newblock {\em A Survey of Quantum Property Testing}.
\newblock Number~7 in Graduate Surveys. Theory of Computing Library, 2016.

\bibitem[MTZ20]{improved_kDist}
Nikhil~S. Mande, Justin Thaler, and Shuchen Zhu.
\newblock Improved approximate degree bounds for k-distinctness.
\newblock In {\em Proceedings of the 15th Conference on the Theory of Quantum Computation, Communication and Cryptography (TQC)}, volume 158, pages 2:1--2:22. Schloss Dagstuhl -- Leibniz-Zentrum f{\"u}r Informatik, 2020.

\bibitem[OR11]{orenstein2011testing}
Yaron Orenstein and Dana Ron.
\newblock Testing {E}ulerianity and connectivity in directed sparse graphs.
\newblock {\em Theoretical Computer Science}, 412(45):6390--6408, 2011.

\bibitem[PW23]{peng2023optimal}
Pan Peng and Yuyang Wang.
\newblock An optimal separation between two property testing models for bounded degree directed graphs.
\newblock In {\em 50th International Colloquium on Automata, Languages, and Programming (ICALP)}, volume 261, pages 96:1--96:16. Schloss Dagstuhl -- Leibniz-Zentrum f{\"u}r Informatik, 2023.

\bibitem[Ron10]{DBLP:journals/fttcs/Ron09}
Dana Ron.
\newblock Algorithmic and analysis techniques in property testing.
\newblock {\em Foundations and Trends in Theoretical Computer Science}, 5(2):73--205, 2010.

\bibitem[RRSS09]{raskhodnikova2009strong}
Sofya Raskhodnikova, Dana Ron, Amir Shpilka, and Adam Smith.
\newblock Strong lower bounds for approximating distribution support size and the distinct elements problem.
\newblock {\em SIAM Journal on Computing}, 39(3):813--842, 2009.

\bibitem[RS11]{DBLP:journals/siamdm/RubinfeldS11}
Ronitt Rubinfeld and Asaf Shapira.
\newblock Sublinear time algorithms.
\newblock {\em SIAM Journal on Discrete Mathematics}, 25(4):1562--1588, 2011.

\bibitem[She11]{Sherstov11}
Alexander~A. Sherstov.
\newblock The pattern matrix method.
\newblock {\em SIAM Journal on Computing (SICOMP)}, 40(6):1969--2000, 2011.

\bibitem[She13]{Sher13}
Alexander~A. Sherstov.
\newblock The intersection of two halfspaces has high threshold degree.
\newblock {\em SIAM Journal on Computing}, 42(6):2329--2374, 2013.

\bibitem[SZ09]{ShiZ09}
Yaoyun Shi and Yufan Zhu.
\newblock Quantum communication complexity of block-composed functions.
\newblock {\em Quantum Information \& Computation}, 9(5):444--460, 2009.

\bibitem[YI10a]{yoshida2010query}
Yuichi Yoshida and Hiro Ito.
\newblock Query-number preserving reductions and linear lower bounds for testing.
\newblock {\em IEICE transactions on Information and Systems}, E93.D(2):233--240, 2010.

\bibitem[YI10b]{yoshida2010testing}
Yuichi Yoshida and Hiro Ito.
\newblock Testing k-edge-connectivity of digraphs.
\newblock {\em Journal of Systems Science and Complexity}, 23(1):91--101, 2010.

\end{thebibliography}
